\pdfoutput=1

\documentclass[11pt]{article}

\PassOptionsToPackage{dvipsnames,usename}{xcolor}
\usepackage[preprint]{acl}

\usepackage{times}
\usepackage{latexsym}

\usepackage[T1]{fontenc}
\usepackage[utf8]{inputenc}

\usepackage{microtype}

\usepackage{inconsolata}

\usepackage{graphicx}

\usepackage{xspace}
\usepackage{amsfonts}
\usepackage{amsthm}
\usepackage{amsmath}
\usepackage{amssymb}
\usepackage{cleveref}
\usepackage{bbm}
\usepackage{caption}
\usepackage{subcaption}
\usepackage{xcolor}
\usepackage{tikz}
\usetikzlibrary{bayesnet}
\usepackage{thm-restate}
\usepackage{booktabs}
\usepackage[most]{tcolorbox}
\usepackage{enumitem}
\usepackage{wasysym}
\usepackage{multirow}

\DeclareFontFamily{OT1}{pzc}{}
\DeclareFontShape{OT1}{pzc}{m}{it}{<-> s * [1.10] pzcmi7t}{}
\DeclareMathAlphabet{\mathpzc}{OT1}{pzc}{m}{it}

\newtheorem{defin}{Definition}
\newtheorem{lemma}{Lemma}
\newtheorem{theorem}{Theorem}

\newtheorem{reduction}{Reduction}

\newtheorem{mylemmastep}{LemmaProofStep}

\makeatletter
\@addtoreset{mylemmastep}{section}
\makeatother

\usepackage{cleveref}
\crefname{section}{\S}{\S\S}
\Crefname{section}{\S}{\S\S}
\crefname{table}{Tab.}{}
\crefname{figure}{Fig.}{Figs.}
\crefname{algorithm}{Alg.}{}
\crefname{equation}{Eq.}{Eqs.}
\crefname{appendix}{App.}{}
\crefname{thm}{Theorem}{}
\crefname{prop}{Proposition}{}
\crefname{defin}{Definition}{}
\crefname{reduction}{Reduction}{}
\crefname{cor}{Corollary}{}
\crefname{observation}{Observation}{}
\crefname{assumption}{Assumption}{}
\crefname{lemma}{Lemma}{Lemmas}
\crefformat{section}{\S#2#1#3}
\crefformat{footnote}{#2\footnotemark[#1]#3}

\DeclareMathOperator*{\BigCirc}{\bigcirc}
\DeclareMathOperator*{\argmin}{\mathrm{arg\,min}}
\DeclareMathOperator*{\argmax}{\mathrm{arg\,max}}

\newcommand{\colourbase}{MidnightBlue}
\newcommand{\coloursubword}{purple}
\newcommand{\colourcharacter}{ForestGreen}
\newcommand{\colourmerge}{purple}
\newcommand{\coloursat}{black}
\newcommand{\colourtok}{black}

\newcommand{\mymacro}[2]{\newcommand{#1}{{\color{\colourbase}#2}}}
\newcommand{\mytemp}[2]{\newcommand{#1}{{#2}}}
\newcommand{\mymerge}[2]{\newcommand{#1}{{\color{\colourmerge}#2}}}
\newcommand{\mysat}[2]{\newcommand{#1}{{\color{\coloursat}#2}}}
\newcommand{\mytok}[2]{\newcommand{#1}{{\color{\colourtok}#2}}}
\newcommand{\mysubword}[2]{\newcommand{#1}{{\color{\coloursubword}#2}}}
\newcommand{\mycharacter}[2]{\newcommand{#1}{{\color{\colourcharacter}#2}}}

\newcommand{\defn}[1]{\textbf{#1}}
\mycharacter{\character}{c}
\mycharacter{\characters}{\mathbf{\character}}
\mycharacter{\alphabet}{\Sigma}
\mysubword{\subword}{s}
\mysubword{\subwords}{\mathbf{\subword}}
\mysubword{\vocab}{\mathcal{S}}
\mymacro{\tokenise}{\mathtt{tok}}
\mymacro{\concat}{\mathtt{concat}}

\mytok{\dataset}{\mathcal{D}}
\mytok{\vocabsize}{K}
\mytok{\maxsymbols}{\delta}

\mytemp{\nrepeatvar}{f}
\mytemp{\satnvariables}{J}
\mytemp{\satnclauses}{I}
\mytemp{\vocabvariables}{\vocab_{0}}

\mysat{\minclauses}{\psi}
\mysat{\satvar}{X}
\mysat{\satval}{x}
\newcommand{\satyesvarj}{x^{\mathtt{T}}_j}
\newcommand{\satyesvarjprime}{x^{\mathtt{T}}_{j'}}
\newcommand{\satnotvarj}{x^{\mathtt{F}}_j}
\newcommand{\satnotvarjprime}{x^{\mathtt{F}}_{j'}}
\mysat{\satvars}{\mathcal{X}}
\mysat{\satvals}{\mathpzc{x}}
\mysat{\satclauses}{\mathcal{L}}
\mysat{\satliteral}{L}
\mysat{\satliteralval}{\ell}
\mysat{\sateq}{\varphi}
\newcommand{\spacesymbol}{\circledcirc}
\newcommand{\spacesymboltwo}{\otimes}

\newcommand{\one}{\mathbbm{1}}
\newcommand{\valfalse}{\mathtt{F}}
\newcommand{\valtrue}{\mathtt{T}}

\mymerge{\merge}{m}
\mymerge{\merges}{\mathbf{m}}
\mymerge{\mergeset}{\mathcal{M}}
\mymacro{\compression}{\Delta}
\mymacro{\mincompression}{\Delta_{\mathrm{min}}}
\mymacro{\tokeniselengthfun}{\mathtt{toklen}}
\newcommand{\tokeniselength}{\tokeniselengthfun(\dataset, \vocab)}
\newcommand{\tokeniselengthmerge}{\tokeniselengthfun(\dataset, \merges)}
\mymacro{\vocabbasemerge}{\vocab_{0}}
\mymacro{\mergebase}{\merges_{0}}
\mymacro{\mergefunc}{\mathtt{merge}}
\mymacro{\emptystring}{\emptyset}
\newcommand{\bigfunccomp}{\bigodot}

\mymacro{\directtoken}{\tokenise_{\text{\pointer}}}
\mymacro{\directtokenmid}{\directtoken[\vocab]}
\mymacro{\directtokenfull}{\directtoken[\vocab](\characters)}
\mymacro{\bottomuptoken}{\tokenise_{\uparrow}}
\mymacro{\bottomuptokenmid}{\bottomuptoken[\merges]}
\mymacro{\bottomuptokenfull}{\bottomuptoken[\merges](\characters)}

\mymacro{\reductionfunc}{\mathrm{R1}}
\newcommand{\reductionfuncfull}{\reductionfunc(\satvars, \satclauses, \minclauses)}
\mymacro{\reductiontwofunc}{\mathrm{R2}}
\newcommand{\reductiontwofuncfull}{\reductiontwofunc(\satvars, \satclauses, \minclauses)}
\mymacro{\maxsat}{\mathrm{M2S}}
\newcommand{\maxsatfull}{\maxsat(\satvars, \satclauses, \minclauses)}
\mymacro{\mintok}{\mathrm{Tok}_{\text{\pointer}}}
\newcommand{\mintokfull}{\mintok(\dataset, \vocabsize, \maxsymbols)}
\mymacro{\mintokmerge}{\mathrm{Tok}_{\uparrow}}

\newcommand{\satsatisfiedtag}{\star}
\newcommand{\satvarsatisfiedj}{\satval_j^{\satsatisfiedtag}}
\newcommand{\satvarsatisfiedjprime}{\satval_{j'}^{\satsatisfiedtag}}
\newcommand{\token}[1]{\ensuremath{\color{\coloursubword}#1}} 
\newcommand{\charstring}[1]{{\color{\colourcharacter}#1}}
\newcommand{\subwordstring}[1]{{\color{\coloursubword}#1}}

\mymacro{\bijectionvocabsat}{\mathrm{Conv}_{\vocab\to\satvals}}
\mymacro{\bijectionmergesat}{\mathrm{Conv}_{\merges\to\satvals}}

\newcommand{\mergestring}[2]{\langle\subwordstring{#1},\subwordstring{#2}\rangle}

\newcommand{\mergesthreecharsone}{\merges_{j}^{\spacesymbol}}
\newcommand{\mergesthreecharstwo}{\merges_{j}^{\spacesymboltwo}}
\newcommand{\mergesthreecharstrue}{\merges_{j}^{\valtrue}}
\newcommand{\mergesthreecharsfalse}{\merges_{j}^{\valfalse}}
\newcommand{\mergesthreecharstrueprime}{\merges_{j'}^{\valtrue}}
\newcommand{\mergesthreecharsfalseprime}{\merges_{j'}^{\valfalse}}

\mymacro{\stringequiv}{\stackrel{{\circ}}{=}}
\newcommand{\defeq}{\stackrel{\texttt{\tiny def}}{=}}
\mymacro{\detokenise}{\mathtt{detok}}
\mymacro{\objectivefunc}{\mathfrak{G}}
\newcommand{\btheta}{\boldsymbol{\theta}}
\newcommand{\ptheta}{p_{\btheta}}

\newcommand*{\circled}[1]{\tikz[baseline=(char.base)]{
        \node[shape=circle,draw,inner sep=1pt] (char) {\normalfont{\small #1}};}}

\newcommand{\maxsatacron}{\texttt{max-2-SAT}\xspace}

\title{Tokenisation is NP-Complete}

\newcommand{\makesf}[1]{\textsf{{{#1}}}}
\newcommand{\ethemailadress}[1]{\href{mailto:#1@inf.ethz.ch}{\makesf{#1}}}
\author{Philip Whittington, Gregor Bachmann, Tiago Pimentel \\
  Department of Computer Science, ETH Z\"urich \\
  \makesf{\{}\ethemailadress{philip.whittington},
  \ethemailadress{gregor.bachmann},
  \ethemailadress{tiago.pimentel}\makesf{\}@inf.ethz.ch}
  }

\begin{document}
\maketitle
\begin{abstract}
    In this work, we prove the NP-completeness of two variants of tokenisation, defined as the problem of compressing a dataset to at most $\maxsymbols$ symbols by either finding a vocabulary directly (\emph{direct} tokenisation), or selecting a sequence of merge operations (\emph{bottom-up} tokenisation).
\end{abstract}

\section{Introduction}

Tokenisation is at the heart of natural language processing (NLP) being the first step required to use a language model (LM).
Given a string of characters $\characters$, a tokeniser converts it into a string of \defn{subwords} $\subwords$.
Language models are then trained to estimate distributions over subword strings---never seeing the original character strings.
Despite its prominent role, however, much remains unknown about tokenisation.
We still do not know, for instance, what makes a good tokeniser \citep{gowda-may-2020-finding,cognetta-etal-2024-two}:
which characteristics should its produced subwords $\subwords$ have to be a good starting point for language modelling?
If we knew this, then we could define an \defn{objective function} which we could evaluate tokenisers with.

Another open question is how to---given such an objective function---efficiently find a tokeniser which maximises it.
Byte pair encoding \citep[BPE;][]{gage1994new,sennrich-etal-2016-neural}, for instance, is a greedy solution to find a tokeniser which maximises a text's compression.
UnigramLM \citep{kudo-2018-subword} is a heuristic method to find a tokeniser that maximises its tokenised text's unigram log-probability.
Both these methods, however, are approximate: they do not necessarily find an optimal tokeniser according to their objective function.
This raises the question of whether finding such optimal tokenisers efficiently is even possible.

In this paper, we answer this question (at least partially), proving the NP-completeness of several variants of this tokenisation problem.
Specifically, we focus on finding tokenisers that maximise the \defn{compression} of a text.\footnote{The compression achieved by a tokeniser correlates with downstream language modelling performance \citep{galle-2019-investigating,zouhar-etal-2023-tokenization} and 
computational efficiency.}
Given this objective, we then define the \defn{tokenisation problem} as the task of finding a tokeniser which compresses a dataset to at most $\maxsymbols$ symbols.
Notably, prior work imposes different constraints on how tokenisers are defined; here we consider two variants.
In \defn{direct tokenisation}, the desired compression must be reached by choosing a vocabulary (i.e., a set of subwords) which is directly used to represent the text.
In \defn{bottom-up tokenisation}, the desired compression must be reached by finding a sequence of merge operations instead, which we apply to the input text.\looseness=-1

We prove the NP-hardness of both of these tokenisation problems (as well as of some variants thereof) by reducing from the \defn{max 2-satisfiability} (\maxsatacron) problem.\footnote{Concurrent work by \citet{kozma2024theoretical} also proves the NP-completeness of bottom-up tokenisation. In fact, they prove its APX-hardness.}
Practically speaking, our results imply that we are unlikely to find an efficient algorithm for the problem of finding optimal tokenisers, and that we should focus on approximate algorithms (such as BPE or UnigramLM) instead.\looseness=-1

\section{How to Choose a Tokeniser?}

In theory, a researcher's choice of tokeniser should not influence their language model's quality.
This is because we can extract word-level \citep{pimentel2024compute,oh2024leading} or character-level \citep{phan2024understanding,giulianelli-etal-2024-proper} distributions from any subword-level language model. 
Thus, regardless of which tokeniser is used, a sufficiently expressive language model should be able to represent the exact distributions over characters or words that we are interested in.

In practice, however, a bad choice of tokeniser 
can have undesirable effects on downstream applications.
For instance, performing standard arithmetic tasks (e.g., $317 + 421$) can be difficult even for large models \citep{nogueira2021investigating,muffo-etal-2022-evaluating} due to the arbitrary splitting of numbers into subwords.
Indeed, simple changes in how numbers are tokenised can improve performance in such tasks \citep{singh2024tokenization}.
Similar issues arise when prompting LMs to count letters in words, where even advanced models such as \texttt{GPT-4} 
infamously cannot correctly count the number of occurrences of the letter \charstring{r} in the word \subwordstring{strawberry}.

This raises the question of how to select a good tokeniser.
Ideally, we would choose the tokeniser which maximises downstream language modelling performance.
Unfortunately, we do not know how to measure such performance without fully training a model, making its direct maximisation computationally infeasible.
Rather, we thus optimise proxy objectives---assumed to correlate with downstream performance.
Among these are unigram log-probability \citep{kudo-2018-subword}, R\'enyi efficiency \citep{zouhar-etal-2023-tokenization}, and compression \citep{galle-2019-investigating}.

We focus on compression in this paper. 
Denoting our tokenisation's \textbf{objective function} as $\objectivefunc$, we write this objective as: $\objectivefunc(\subwords) = - |\subwords|$. 
Improved compression leads to: (i) more efficient training and inference, due to shortened inputs;\footnote{Recent work tries to improve the computational efficiency of byte-level models \citep{yu2023megabyte,pagnoni2024byte}.}
(ii) improved downstream performance, at least to a certain extent \citep{galle-2019-investigating,rust-etal-2021-good,zouhar-etal-2023-tokenization,goldman-etal-2024-unpacking};\footnote{Although, see also \citet{ali-etal-2024-tokenizer}, who argue that compression might be a necessary but not sufficient condition for good downstream performance, and \citet{schmidt-etal-2024-tokenization}, who argue that compression and downstream performance have a more complex relationship than prior work suggests.\looseness=-1} and
(iii) fairer multilingual treatment, given models' limited context lengths and the per-token costs to use proprietary models \citep{petrov2023language,ahia2023do}.

\begin{tcolorbox}[colback=white,colframe=gray,left=4pt,title=\!\!{\small Our Notation's Colour-coding}]
    {\small\begin{itemize}[leftmargin=2mm]
        \item {\color{\colourcharacter} Green} for raw data (i.e., characters $\characters \in \alphabet^*$);
        \item {\color{\coloursubword} Purple} for tokeniser-specific data (i.e., subwords $\subwords \!\in\! \vocab^*$ and merges $\merges \!\in\! \mergeset^*$);
        \item {\color{\colourbase} Blue} for functions (e.g., $\tokenise$).
    \end{itemize}}
\end{tcolorbox} 

\section{Defining a Tokeniser} \label{sec:apply_tokenisers}

A tokeniser can be defined as a 3-tuple $\langle\vocab, \tokenise, \detokenise \rangle$, composed of a vocabulary, a tokenisation and a detokenisation function.
Before defining these terms, however, we require some notation.
Let $\characters = \character_1\character_2\cdots\character_{|\characters|} \in \alphabet^*$ be a \defn{character-string}, i.e., a sequence of characters $\character$ from alphabet $\alphabet$.
Further, let $\dataset \!=\! \{\characters_n\}_{n=1}^{N}$ be a dataset of character-strings.\footnote{We note that we use set notation here, but our datasets are actually multisets---datasets can include the same string $\characters$ multiple times. We show that tokenisation is still NP-complete for datasets with no repetitions in \cref{sec:other_tokenisation_problems}. Further, we impose no constraint on the kind of string present in these datasets:  each $\characters_n$ can be either a raw or pre-tokenised character-string (i.e., either a full document or a whitespace word).}
A subword $\subword \!\in\! \vocab$ represents a non-empty character-string $\characters$. (Sequence $\characters$ can have length one.) 
Finally, let $\subwords = \subwordstring{\langle\subword_1, \subword_2, \cdots, \subword_{|\subwords|}\rangle} \in \vocab^*$ be a \defn{subword-string}.
Just like a single subword, a subword-string $\subwords \in \vocab^*$ represents a character-string via the concatenation of its subwords' characters:
\begin{align}
    \concat(\subwords) = \subword_1 \circ \subword_2 \circ ... \circ \subword_{|\subwords|}
\end{align}
and we say that a pair of character and subword strings are equivalent if:
\begin{align}
     \characters \,\stringequiv\, \subwords 
     \iff 
     \characters = \concat(\subwords) 
\end{align}

Given the notation above, we can now define the items in tuple $\langle\vocab, \tokenise, \detokenise \rangle$.
A tokeniser's \defn{vocabulary} is a set of subwords $\vocab \subset \alphabet^*$ such that $\alphabet \subseteq \vocab$; we say its size is $|\vocab| = |\alphabet| + \vocabsize$.
Further, a \defn{detokenisation function} is defined as $\detokenise: \vocab^* \to \alphabet^*$ 
and given a subword-string it outputs the character-string it represents.
This function thus is simply defined as $\detokenise(\subwords) \,\defeq\, \concat(\subwords)$.

Finally, we are left with defining a \defn{tokenisation function} $\tokenise: \alphabet^* \to \vocab^*$, which maps from character- to subword-strings.
Notably, these functions always ensure the equivalence $\characters \,\stringequiv\, \subwords$ for $\subwords\mathop{=}\tokenise(\characters)$.
Several tokenisation functions, however, are compatible with this constraint, as given a vocabulary, many subword-strings may be equivalent to the same character-string.
For instance, given 
$\vocab \!=\! \{\subwordstring{a},\subwordstring{c}, \subwordstring{t}, \subwordstring{at}\}$, the string $\characters \!=\! \charstring{\langle c, a, t \rangle}$ could be tokenised as $\subwords = \subwordstring{\langle c, a, t \rangle}$ or as $\subwords = \subwordstring{\langle c, at \rangle}$.
Most researchers define tokenisation functions in one of two ways, which we term direct and bottom-up tokenisation functions here; we define these next.

\subsection{Direct Tokenisation Functions}

In direct tokenisation, a character-string is directly replaced by an optimal subword-string.
To implement this, one must thus first define what \emph{optimal} means; this is done through objective function $\objectivefunc$ which, given a subword-string, returns a score. 
Given a previously chosen vocabulary $\vocab$ (we discuss how to find $\vocab$ in \cref{sec:direct_is_np_complete}), a direct tokenisation function then encodes string $\characters$ as:\looseness=-1%
\begin{align}
    \directtokenfull = 
    &\argmax_{\subwords \in \vocab^*} \objectivefunc(\subwords) \\
    &\mathrm{s.t.}\,\, \subwords \,\stringequiv\, \characters \nonumber
\end{align}
In words, given a vocabulary $\vocab$, function $\directtoken$ returns the optimal subword-string $\subwords \in \vocab^*$ which is equivalent to the input character-string $\characters$.
We then set $\tokenise(\characters) \defeq \directtokenfull$.
Different choices of $\objectivefunc$ recover methods such as UnigramLM \citep{kudo-2018-subword} or PathPiece \citep{schmidt-etal-2024-tokenization}.
Notably, in general, this function is not efficiently computable.\footnote{In fact, \citet{geh-etal-2024-signal} shows that it is NP-complete for $\objectivefunc(\subwords) \!=\! \sum_{t=1}^{|\subwords|} \log \ptheta(\subword_t \!\mid\! \subwords_{<t})$, where $\ptheta$ is a language model.}

In this paper, we are concerned with tokenisers that use compression as their objective: that is, for which $\objectivefunc(\subwords) = - |\subwords|$. 
In this case, we can rewrite the direct tokenisation function as:
\begin{align} \label{eq:direct_tok_function}
    \directtokenfull = 
    &\argmin_{\subwords \in \vocab^*} |\subwords| \\
    &\mathrm{s.t.}\,\, \subwords \,\stringequiv\, \characters \nonumber
\end{align}
Importantly, in the case of compression, this equation can be computed efficiently (as shown in \cref{sec:direct_is_np}).

\subsection{Bottom-up Tokenisation Functions}
\label{subsec:def_bottomup}

In bottom-up tokenisation, one starts with a set of character-strings, and merges their symbols bottom-up, one pair at a time.\footnote{Currently, this is likely the most common tokenisation function, being used in popular tokenisers such as, e.g., GPT-4's \citep{openai2024gpt4}, LLaMA's \citep{touvron2023llama,touvron2023llama2}, and Pythia's \citep{biderman-etal-2023-pythia}.}
Formally, let $\merge \in \mergeset$ be a \defn{merge}, defined as a pair of subwords: $\merge = \langle \subword_1, \subword_2 \rangle$.
Further, let $\mergeset \defeq \alphabet^* \mathop{\times} \alphabet^*$.
Now, let $\mergefunc$ be a functional; given merge $\merge = \langle \subword_1, \subword_2 \rangle$, it returns a function  $\mergefunc[\merge]: \vocab^* \to (\vocab \cup \{\subword_1 \circ \subword_2\})^*$ which operates on string $\subwords$ left-to-right, replacing every occurrence of $\subword_1$ followed by $\subword_2$ in it with subword $\subword' = \subword_1\circ\subword_2$.
E.g., given $\subwords = \subwordstring{\langle wo,r,ld \rangle}$ and $\merge = \mergestring{wo}{r}$, the output of
$\mergefunc[\merge](\subwords)$ is $\subwordstring{\langle wor,ld \rangle}$.

Consider now $\merges \in \mergeset^*$, a sequence of merges.
Given a character-string $\characters \in \alphabet^*$, a bottom-up tokenisation function compresses it as:
\begin{align}
	\bottomuptokenfull = \bigg(\bigfunccomp_{z=1}^{|\merges|} \mergefunc[\merge_z]\bigg) (\characters)
\end{align}
where $\bigfunccomp$ represents function composition, e.g., 
$\bigfunccomp_{z=1}^{2} \mergefunc[\merge_z] = \mergefunc[\merge_2] \odot \mergefunc[\merge_1]$.
Bottom-up tokenisers then set $\tokenise \!\defeq\! \bottomuptokenmid$.
Further, a merge sequence $\merges$ is also used to set a bottom-up tokeniser's vocabulary as:
\begin{align}\label{eq:bottomup_vocab_defn}
    \vocab = \alphabet\cup\{\subword_1 \circ \subword_2 \mid \mergestring{\subword_1}{\subword_2} \in \merges\}
\end{align}
where $|\merges| = \vocabsize$ implies this vocabulary has size $|\vocab| = |\alphabet| + \vocabsize$, as before.

\section{Maximum 2-Satisfiability}

Our paper's goal is to prove the NP-completeness of tokenisation.
To show this, we must reduce an NP-hard problem to tokenisation in polynomial time.
We will rely on the \defn{maximum 2-satisfiability} problem (\maxsatacron) for this, whose definition we provide here.
The NP-hardness of \maxsatacron was proven by \citet{garey1974maxsat}.

\begin{defin}\label{defn:max2set_decision_problem}
    Let $\satvars = \{\satvar_j\}_{j=1}^{\satnvariables}$ be a set of variables; 
    each of these variables are assigned values $\satval_j \in \{\valfalse,\valtrue\}$, and we write $\satvals = \{\satval_j\}_{j=1}^{\satnvariables} \in \{\valfalse,\valtrue\}^{\satnvariables}$.
    Let $\satclauses = \{(\satliteral_i^1 \lor \satliteral_i^2)\}_{i=1}^{\satnclauses}$ be a set of clauses,\footnote{\maxsatacron also allows clauses to have a single literal $\satliteral_i$. In this case, we can always rewrite the clause as $(\satliteral_i \lor \satliteral_i)$ with no change to the solution of this decision problem.}
    where each literal $\satliteral$ represents either a variable $\satvar_j$ or its negation  $\neg \satvar_j$.
    The \defn{\maxsatacron decision problem} requires deciding whether there exists an assignment for which at least $\minclauses$ clauses are satisfied:\looseness=-1%
    \begin{align}\label{eq:max2sat_satisfiable}
        \minclauses \leq 
        &\max_{\satvals \in \{\valfalse,\valtrue\}^{J}} \sum_{i=1}^{I} \one\{\satliteral_i^1 \lor \satliteral_i^2\} 
    \end{align}
    where $\one$ is an indicator function which evaluates the clause and returns one if the clause is satisfied by $\satvals$ and zero otherwise.
\end{defin}

For mathematical convenience, we will write $\maxsatfull$ for a function which returns $\valtrue$ if its input is satisfiable under a \maxsatacron decision problem, and $\valfalse$ otherwise.
As a concrete example, consider the set of variables $\satvars = \{\satvar_1, \satvar_2\}$ and the set of clauses 
$\satclauses = \{\satvar_1 \mathop{\lor} \satvar_2, \neg \satvar_1 \mathop{\lor} \satvar_2, \satvar_1 \mathop{\lor} \neg \satvar_2, \neg \satvar_1 \mathop{\lor} \neg \satvar_2\}$.
The assignment $\satval_1 \mathop{=} \valtrue, \satval_2 \mathop{=} \valtrue$ leads to $3$ clauses being satisfied, which is the optimum.
For this example, we thus have that $\maxsat(\satvars, \satclauses, 3) = \valtrue$, but that $\maxsat(\satvars, \satclauses, 4) = \valfalse$.

\section{Finding an Optimal Direct Tokeniser}
\label{sec:direct_is_np_complete}

We are now left with the task of finding an optimal tokeniser.
We do this by selecting either:
its vocabulary in direct tokenisation, since $\tokenise = \directtokenmid$; or its merge sequence in bottom-up tokenisation, since $\tokenise = \bottomuptokenmid$ and since its vocabulary is chosen according to \cref{eq:bottomup_vocab_defn}.
(Note that in \cref{sec:apply_tokenisers}, we only showed how to apply tokenisers at inference time, but not how to find them.)
In this section, we focus on direct tokenisation, defining its optimisation and decision problems; we then prove its NP-completeness.
The optimisation problem is defined as follows.

\begin{defin} \label{defn:token_optim_problem}
    Given a dataset $\dataset$ and a vocabulary size $\vocabsize$, the \defn{direct tokenisation optimisation problem} is to find a vocabulary $\vocab \subset \alphabet^*$ which maximally compresses $\dataset$:
    \begin{align}
        \vocab^{\star} = 
        &\argmin_{\vocab \subset \alphabet^*} \sum_{\characters \in \dataset} |\directtokenfull| \\
        & \mathrm{s.t.}\,\,|\vocab| = |\alphabet| + \vocabsize \nonumber
    \end{align}
\end{defin}

We can similarly define direct tokenisation's decision problem.
\begin{defin} \label{defn:token_decision_problem}
    Given a dataset $\dataset$ and a vocabulary size $\vocabsize$, the
    \defn{direct tokenisation decision problem} requires deciding whether there exists a vocabulary $\vocab \subset \alphabet^*$ which compresses $\dataset$ to at most $\maxsymbols$ symbols:
    \begin{align}\label{eq:direct_satisfiable}
        \maxsymbols \geq 
        &\min_{\vocab \subset \alphabet^*} \sum_{\characters \in \dataset} |\directtokenfull| \\
        & \mathrm{s.t.}\,\,|\vocab| = |\alphabet| + \vocabsize \nonumber
    \end{align}
\end{defin}

We write $\mintokfull$ for a function which returns $\valtrue$ if a tokenisation decision problem with those inputs is satisfiable, and $\valfalse$ otherwise.
Note that, whenever $|\dataset| \leq \vocabsize$, the solution to the problem above is trivial, as an optimal solution simply requires including all strings $\characters_n$ in vocabulary $\vocab$.
As we show next, however, in the general case the above decision problem is NP-complete.
We now state this as a theorem, which we will prove in the next two sections.

\begin{theorem}
    The direct tokenisation decision problem, as in \cref{defn:token_decision_problem}, is NP-complete.
\end{theorem}
\begin{proof}
    A decision problem is considered to be NP-complete if: (i) it is in NP; (ii) it is NP-hard.
    We prove these conditions in \cref{sec:direct_is_np} and \cref{sec:direct_is_nphard}.
\end{proof}

\subsection{Direct Tokenisation is in NP} \label{sec:direct_is_np}

A decision problem is in the nondeterministic polynomial time class (NP) if, given a \defn{certificate}
of polynomial length, one can verify that certificate in polynomial time.
Specifically, a certificate usually encodes a decision problem's solution, allowing us to verify its satisfiability.
In the case of direct tokenisation, this certificate would be a vocabulary $\vocab$ which leads a dataset $\dataset$ to be compressed to at most $\maxsymbols$ symbols.
Verifying this certificate simply requires computing the sum in \cref{eq:direct_satisfiable}, i.e.:
\begin{align} \label{eq:direct_satisfiable_proof_check}
    \sum_{\characters \in \dataset} |\directtokenfull|
\end{align}

\begin{lemma}\label{lemma:direct_is_np}
    The direct tokenisation decision problem, as in \cref{defn:token_decision_problem}, is in NP.
\end{lemma}
\begin{proof}
    As noted above, whenever $|\dataset| \leq \vocabsize$, each $\characters_n \in \dataset$ can be included in the vocabulary $\vocab$ and fully compressed to a single symbol; we can thus verify the problem's satisfiability by simply checking that $\maxsymbols \geq |\dataset|$ as this is the best reachable compression.
    Assuming $\vocabsize$ to be bounded by $|\dataset|$---and therefore polynomial in the input---we have that the certificate $\vocab$ also has polynomial length.
    Given such a certificate $\vocab$, verifying it simply requires computing the sum in \cref{eq:direct_satisfiable_proof_check}.
    In turn, computing this sum requires $|\dataset|$ calls to function $\directtoken$.
    It follows that, if function $\directtoken$ runs in polynomial time, then direct tokenisation is in NP.
    Luckily, this function can indeed be computed efficiently using \citeposs{schmidt-etal-2024-tokenization} PathPiece method, which runs in $O(|\characters|^2)$ time.
\end{proof}

\subsection{Direct Tokenisation is NP-hard} \label{sec:direct_is_nphard}

We now use a reduction from \maxsatacron to prove the NP-hardness of direct tokenisation.

\begin{reduction}
    \label{reduction:max2set_to_tokenisation_reduction}
    Let us have a \maxsatacron decision problem defined as in \cref{defn:max2set_decision_problem}.
    To reduce this problem to a tokenisation decision problem, as in \cref{defn:token_decision_problem}, we start by defining an alphabet 
    $\alphabet = \{\charstring{\spacesymbol}\} \cup \{\charstring{\satyesvarj}, \charstring{\satnotvarj}\}_{j=1}^{J}$.
    We then construct three sets of strings:
    \begin{subequations}
    \begin{align}
        &\dataset_1 = 
        \{\charstring{\spacesymbol\satyesvarj\spacesymbol}\}_{j=1}^{\satnvariables} \cup  
        \{\charstring{\spacesymbol\satnotvarj\spacesymbol}\}_{j=1}^{\satnvariables} \\
        &\dataset_2 = 
        \{\charstring{\spacesymbol\satyesvarj\spacesymbol\satnotvarj\spacesymbol}\}_{j=1}^{\satnvariables} \\
        &\dataset_3 = 
        \{\charstring{\spacesymbol}\satliteral_i^1\charstring{\spacesymbol}\satliteral_i^2\charstring{\spacesymbol}\}_{i=1}^{\satnclauses}
    \end{align}
    \end{subequations}
    in these strings $\satliteral_i$ 
    is replaced by either character $\charstring{\satyesvarj}$ or $\charstring{\satnotvarj}$, depending on whether it represents $\satvar_j$ or $\neg\satvar_j$, respectively.
    We then construct our dataset $\dataset$, and choose $\vocabsize$ and $\maxsymbols$ as:%
    \begin{subequations}
    \begin{align}
        &\dataset = \bigg(\bigcup_{\_=1}^{\nrepeatvar} \dataset_1\bigg) \cup \bigg(\bigcup_{\_=1}^{\nrepeatvar'} \dataset_2\bigg) \cup \dataset_3 \\
        &\vocabsize = \satnvariables,\quad
        \maxsymbols = (4\nrepeatvar + 3\nrepeatvar')\,\satnvariables + 5\,\satnclauses - 2\minclauses
    \end{align}
    \end{subequations}
    where we set $\nrepeatvar' \defeq 4\satnclauses + 1$ and $\nrepeatvar \defeq 4\nrepeatvar' \satnvariables + 4\satnclauses + 1$.
\end{reduction}

We write $\reductionfuncfull$ to represent a function which, given an instance of \maxsatacron, returns an instance of the tokenisation problem given by our reduction (i.e., $\dataset, \vocabsize, \maxsymbols$).
For our reduction to be correct, we must have that:
\begin{align}\label{eq:direct_reduction_iff}
    &\maxsatfull \iff \mintok(\reductionfuncfull)
\end{align}
meaning that a \maxsatacron problem is satisfiable if and only if its reduced direct tokenisation problem is as well.
We now set out to prove this.
We start by proving the forward direction of this iff clause.

\begin{restatable}{lemma}{directnphardiflemma} 
\label{lemma:direct_nphard_if_lemma}
    If a \maxsatacron instance is satisfiable, then the direct tokenisation instance output by \cref{reduction:max2set_to_tokenisation_reduction} is also satisfiable. Formally:
    \begin{align}
        &\maxsatfull \implies \mintok(\reductionfuncfull)
    \end{align}
\end{restatable}
\begin{proof}[Proof sketch]
    See a formal proof in \cref{appendix:proof_direct_nphard_if_lemma}.
    Our proof works by first fixing a satisfying solution to \maxsatacron 
    with values $\satvarsatisfiedj$. 
    Given this solution, for each variable, we add to our vocabulary $\vocab$ a subword \token{\spacesymbol\satyesvarj\spacesymbol} if $\satvarsatisfiedj$ is true, or \token{\spacesymbol\satnotvarj\spacesymbol} if $\satvarsatisfiedj$ is false.
    Given these subwords, strings in $\dataset_1$ and $\dataset_2$ occupy a total length of $(4\nrepeatvar + 3\nrepeatvar')\,\satnvariables$.
    Further, since at least $\minclauses$ of the \maxsatacron problem are satisfied by $\satvarsatisfiedj$, the strings in $\dataset_3$ will occupy a total length smaller or equal to $5\,\satnclauses - 2\minclauses$.
    This solution to the tokenisation problem thus gives us a total length which is smaller or equal to $\maxsymbols = (4\nrepeatvar + 3\nrepeatvar')\,\satnvariables + 5\,\satnclauses - 2\minclauses$.
\end{proof}

Now, we are left with proving the backward direction of the iff clause in \cref{eq:direct_reduction_iff}.
We do so in the following lemma.

\begin{restatable}{lemma}{directnphardonlyiflemma} 
\label{lemma:direct_nphard_onlyif_lemma}
    If the direct tokenisation instance output by \cref{reduction:max2set_to_tokenisation_reduction} is satisfiable, the \maxsatacron instance reduced to it is as well. Formally:
    \begin{align}
        &\mintok(\reductionfuncfull) \implies \maxsatfull
    \end{align}
\end{restatable}
\begin{proof}[Proof sketch]
    See a formal proof in \cref{appendix:proof_direct_nphard_onlyif_lemma}.
    Our proof works in three steps.
    First, we show that all satisfying solutions 
    must only have subwords of the form \subwordstring{$\spacesymbol\satyesvarj\spacesymbol$} or \subwordstring{$\spacesymbol\satnotvarj\spacesymbol$}, since this is required to compress strings in $\dataset_1$ to at most $4\nrepeatvar \satnvariables$ symbols.
    Second, we show that all satisfying solutions 
    must only have either subword \subwordstring{$\spacesymbol\satyesvarj\spacesymbol$} or \subwordstring{$\spacesymbol\satnotvarj\spacesymbol$} for any variable $\satvar_j$; this is required to compress strings in $\dataset_2$ to at most $3\nrepeatvar' \satnvariables$ symbols.
    Finally, we show that if a tokeniser compresses strings in $\dataset_3$ to $5\satnclauses - 2\minclauses$, then there is an assignment $\satvals$ which satisfies at least $\minclauses$ of the original \maxsatacron problem.
\end{proof}

Given both lemmas above, we can now trivially prove that direct tokenisation is NP-hard.

\begin{lemma} \label{lemma:direct_is_nphard}
    The direct tokenisation decision problem, as in \cref{defn:token_decision_problem}, is NP-hard.
\end{lemma}
\begin{proof}
    First, it is easy to see that \cref{reduction:max2set_to_tokenisation_reduction} runs in polynomial time.
    Second, \maxsatacron is an NP-hard problem \cite{garey1974maxsat}.
    This lemma then follows trivially from \cref{lemma:direct_nphard_if_lemma,lemma:direct_nphard_onlyif_lemma}, which together show that an instance of the tokenisation problem generated through \cref{reduction:max2set_to_tokenisation_reduction} is satisfiable if and only if the \maxsatacron instance used to produce it is also satisfiable.
\end{proof}

\section{Finding Optimal Bottom-up Tokenisers}

In this section, we shift our attention to bottom-up tokenisation.
We define both its optimisation and decision problems, and then prove its NP-completeness.
We start with defining the optimisation problem.

\begin{defin} \label{defn:bottomup_optim_problem}
    Given a dataset $\dataset$ and a vocabulary size $\vocabsize$, the \defn{bottom-up tokenisation optimisation problem} is to find a merge sequence $\merges \subset \mergeset^*$ which maximally compresses $\dataset$:
    \begin{align}
        \merges^{\star} = 
        &\argmin_{\merges \subset \mergeset^*} \sum_{\characters \in \dataset} |\bottomuptokenfull| \\
        & \mathrm{s.t.}\,\,|\merges| = \vocabsize \nonumber
    \end{align}
\end{defin}
As can be seen, this optimisation problem is similar to the direct tokenisation problem, albeit its target is to find a merge sequence instead of a vocabulary.
We similarly define a decision problem.

\begin{defin} \label{defn:bottomup_decision_problem} 
    Given a dataset $\dataset$ and a vocabulary size $\vocabsize$, the \defn{bottom-up tokenisation decision problem} requires deciding whether there exists a merge sequence $\merges \in \mergeset^*$
    which compresses $\dataset$ to at most $\maxsymbols$ symbols:
    \begin{align}\label{eq:bottomup_tokenisation_satisfiable}
        \maxsymbols \geq 
        &\min_{\merges \in \mergeset^*} \sum_{\characters \in \dataset} |\bottomuptokenfull| \\
        & \mathrm{s.t.}\,\,|\merges| = \vocabsize \nonumber
    \end{align}
\end{defin}

We spend the rest of this section showing that bottom-up tokenisers are NP-complete.

\begin{theorem}
	The bottom-up tokenisation decision problem, as in \cref{defn:bottomup_decision_problem}, is NP-complete.
\end{theorem}
\begin{proof}
    We prove this in two steps below. We first prove that this problem is NP, in \cref{sec:merge_problem_np}. We then prove that this problem is NP-hard, in \cref{sec:merge_problem_nphard}.
\end{proof}

\subsection{Bottom-up Tokenisation is in NP}\label{sec:merge_problem_np} 

We can verify this using a solution, the merge sequence $\merges \in \mergeset^*$, as a certificate.
By showing that this certificate has polynomial length and that it can be verified in polynomial time, we prove this problem is in NP.
To verify this certificate, we simply need to compute the sum in \Cref{eq:bottomup_tokenisation_satisfiable}, i.e.:
\begin{align}\label{eq:bottomup_satisfiable_proof_check}
    \sum_{\characters \in \dataset} |\bottomuptokenfull|
\end{align}
which we show now can be done efficiently.

\begin{lemma}\label{lemma:bottomup_is_np}
    The bottom-up tokenisation decision problem, as in \cref{defn:bottomup_decision_problem}, is in NP.
\end{lemma}
\begin{proof}
    First, if $\vocabsize$ is larger than the total number of characters in $\dataset$, i.e., $\sum_{\characters \in \dataset} |\characters|$,
    then this dataset can be compressed to $|\dataset|$ by merging each string down to a single symbol; further, compressing $\dataset$ more than that is not possible independently of $\vocabsize$.
    Verifying the satisfiability of such an instance of the tokenisation problem is thus trivial, only requiring checking if $\maxsymbols \geq |\dataset|$.
    Second, if $\vocabsize$ is bounded by $|\dataset|$---and therefore polynomial in the input---the certificate $\merges$ has polynomial length.
    Given such a certificate $\merges$, verifying it then simply requires computing the sum in \cref{eq:bottomup_satisfiable_proof_check}.
    In turn, computing this sum requires $|\dataset|$ calls to function $\bottomuptoken$.
    It follows that, if function $\bottomuptoken$ runs in polynomial time, then bottom-up tokenisation is in NP.
    The computation of $\bottomuptoken$, can be done in polynomial time following the structure described in \Cref{subsec:def_bottomup}.
    For each $\merge = \langle \subword_1, \subword_2 \rangle \in \merges$, scan the current $\characters$ and replace each occurrence of $\subword_1, \subword_2$ by $\subword'$.
    This takes time $\mathcal{O}(|\characters|)$ for each merge.
    Afterwards, the resulting string can be compared against the desired size.
    We obtain a total runtime of $O(|\dataset| |\characters| |\merges|)$.
\end{proof}

\subsection{Bottom-up Tokenisation is NP-hard}\label{sec:merge_problem_nphard}

As before, we use a reduction from \maxsatacron to prove bottom-up tokenisation's NP-hardness.

\begin{reduction}
    \label{reduction:max2set_to_bottomup_reduction}
    Let us have a \maxsatacron decision problem defined as in \cref{defn:max2set_decision_problem}.
    To reduce this problem to a bottom-up tokenisation decision problem, as in \cref{defn:bottomup_decision_problem}, we start by defining an alphabet $\alphabet = \{\charstring{\spacesymbol}, \charstring{\spacesymboltwo}\} \cup \{\charstring{\satyesvarj}, \charstring{\satnotvarj}\}_{j=1}^{\satnvariables}$.
    We then construct five sets of strings:%
    {\allowdisplaybreaks
    \begin{align}\label{eq:dataset_reduction_bottomup}
        &\!\!\!\! \dataset_1 \!=\! \{ {\small\charstring{\spacesymbol\satyesvarj}}\}_{j=1}^{\satnvariables} \!\cup\!  \{ {\small\charstring{\satnotvarj\spacesymbol}}\}_{j=1}^{\satnvariables}
        \!\cup\! \{ {\small\charstring{\satyesvarj\spacesymbol}}\}_{j=1}^{\satnvariables} \\
        &\qquad\qquad \cup  \{ {\small\charstring{\spacesymbol\satnotvarj}}\}_{j=1}^{\satnvariables}
        \!\cup\! \{ {\small\charstring{\satyesvarj\spacesymboltwo}}\}_{j=1}^{\satnvariables} \!\cup\!  \{ {\small\charstring{\spacesymboltwo\satnotvarj}}\}_{j=1}^{\satnvariables} 
        \!\!\!\!\!\! \nonumber\\
        &\!\!\!\! \dataset_2 \!=\! \{ {\small\charstring{\spacesymbol\satyesvarj\spacesymbol}}\}_{j=1}^{\satnvariables} \cup  \{ {\small\charstring{\spacesymbol\satnotvarj\spacesymbol}}\}_{j=1}^{\satnvariables} \nonumber \\
        & \qquad\qquad \cup \{ {\small\charstring{\spacesymbol\satyesvarj\spacesymboltwo}}\}_{j=1}^{\satnvariables} \cup  \{ {\small\charstring{\spacesymboltwo\satnotvarj\spacesymbol}}\}_{j=1}^{\satnvariables} \nonumber \\
        &\!\!\!\! \dataset_3 \!=\! \{ {\small\charstring{\spacesymbol\satyesvarj\spacesymbol\satnotvarj\spacesymbol}}\}_{j=1}^{\satnvariables} 
        \!\cup\! \{ {\small\charstring{\spacesymboltwo\satnotvarj\spacesymbol\satyesvarj\spacesymboltwo}}\}_{j=1}^{\satnvariables} \!\!\!\! \nonumber \\
        &\!\!\!\! \dataset_4 \!=\! \{ {\small\charstring{\spacesymbol\satnotvarj\spacesymbol\satyesvarj\spacesymboltwo}} \}_{j=1}^{\satnvariables} 
        \!\cup\! \{ {\small\charstring{\spacesymboltwo\satnotvarj\spacesymbol\satyesvarj\spacesymbol}} \}_{j=1}^{\satnvariables} \!\!\! \nonumber \\
        &\!\!\!\! \dataset_5 \!=\! \left\{{\small\begin{array}{lr}
            \!\!\!\!\charstring{\spacesymbol\satyesvarj\spacesymbol\satnotvarjprime\spacesymbol} \!\!\!\! & 
            \mathtt{if}\,\, \satliteral_i^1 = \satvar_{\!j} \,\,\mathtt{and}\,\,\satliteral_i^2 = \neg\satvar_{\!j'}  \!\!\!\! \\
            \!\!\!\!\charstring{\spacesymbol\satyesvarjprime\spacesymbol\satnotvarj\spacesymbol} \!\!\!\! & 
            \mathtt{if}\,\, \satliteral_i^1 = \neg\satvar_{\!j} \,\,\mathtt{and}\,\,\satliteral_i^2 = \satvar_{\!j'}  \!\!\!\! \\
            \!\!\!\!\charstring{\spacesymboltwo\satnotvarj\spacesymbol\satnotvarjprime\spacesymbol} \!\!\!\! & 
            \mathtt{if}\,\, \satliteral_i^1 = \neg\satvar_{\!j} \,\,\mathtt{and}\,\,\satliteral_i^2 = \neg\satvar_{\!j'}  \!\!\!\! \\
            \!\!\!\!\charstring{\spacesymbol\satyesvarj\spacesymbol\satyesvarjprime\spacesymboltwo} \!\!\!\! & 
            \mathtt{if}\,\, \satliteral_i^1 = \satvar_{\!j} \,\,\mathtt{and}\,\,\satliteral_i^2 = \satvar_{\!j'} \!\!\!\!
        \end{array}}
        \right\}_{i=1}^{\satnclauses} \!\!\!\!\!\!\nonumber
    \end{align} }
    We then construct our dataset $\dataset$, and choose $\vocabsize$ and $\maxsymbols$ as:%
    \begin{align}
        &\dataset = 
        \bigcup_{\_=1}^{\nrepeatvar} \!\dataset_1 \!\cup\! \bigcup_{\_=1}^{\nrepeatvar'} \!\dataset_2 \!\cup\! \bigcup_{\_=1}^{\nrepeatvar''} \!\dataset_3\!\cup\! \bigcup_{\_=1}^{\nrepeatvar'''} \!\dataset_4 \!\cup\! \dataset_5  \\
        &\vocabsize = 8\satnvariables,\,\,
        \maxsymbols = (6\nrepeatvar \!+\! 6\nrepeatvar' \!+\! 4\nrepeatvar'' \!+\! 4\nrepeatvar''')\,\satnvariables \!+\! 3\,\satnclauses \!-\! \minclauses 
        \nonumber
    \end{align}
    where we set:
    \begin{subequations}
    \begin{align}
        \nrepeatvar''' &\defeq 5\satnclauses, \quad 
        \nrepeatvar'' \defeq 10\nrepeatvar'''\satnvariables + 5\satnclauses \\
        \nrepeatvar' &\defeq (10\nrepeatvar'' + 10\nrepeatvar''')\,\satnvariables + 5\satnclauses \\
        \nrepeatvar &\defeq (12\nrepeatvar' + 10\nrepeatvar'' + 10\nrepeatvar''')\,\satnvariables + 5\satnclauses
    \end{align}
    \end{subequations}
\end{reduction}

As before, we write $\reductiontwofuncfull$ for a function which, given the inputs of a \maxsatacron problem, returns the parameters of a bottom-up tokenisation problem.
For our reduction to be correct, we must have that:
\begin{align}\label{eq:bottomup_reduction_iff}
    \maxsatfull \iff \mintokmerge(\reductiontwofuncfull)
\end{align}
We follow the same proof strategies as before, starting by proving the forward direction of this iff statement.

\begin{restatable}{lemma}{bottomupnphardiflemma} 
\label{lemma:bottomup_nphard_if_lemma}
    If a \maxsatacron instance is satisfiable, then the bottom-up tokenisation instance output by \cref{reduction:max2set_to_bottomup_reduction} is also satisfiable. Formally:
    \begin{align}
        &\maxsatfull \implies \mintokmerge(\reductiontwofuncfull)
    \end{align}
\end{restatable}
\begin{proof}[Proof sketch]
    See a formal proof in \cref{appendix:proof_bottomup_nphard_if_lemma}.
    Without loss of generality, let a satisfying solution to \maxsatacron have values $\satvarsatisfiedj$. 
    Our proof works by first defining the three following lists of merges, which must be included in any satisfying solution to the tokenisation problem:
    \begin{subequations} \label{eq:bottomup_merges_always}
    \begin{align}
        \merges_1 = \bigcirc_{j=1}^{\satnvariables} [\mergestring{\spacesymboltwo}{\satnotvarj}, \mergestring{\satyesvarj}{\spacesymboltwo}] \\
        \merges_3 = \bigcirc_{j=1}^{\satnvariables} [\mergestring{\satnotvarj}{\spacesymbol}, \mergestring{\spacesymbol}{\satyesvarj}] \\
        \merges_5 = \bigcirc_{j=1}^{\satnvariables} [\mergestring{\spacesymbol}{\satnotvarj}, \mergestring{\satyesvarj}{\spacesymbol}]
    \end{align}
    \end{subequations}
    We then construct two other lists of merges, which do depend on the satisfying assignments to \maxsatacron:
    \begin{subequations} \label{eq:bottomup_merges_cond}
    \begin{align}
        \merges_2 = \bigcirc_{j=1}^{\satnvariables} \left[
        \begin{array}{lr}
            \mergestring{\spacesymbol}{\satyesvarj\spacesymboltwo} & \mathtt{if}\,\,\satvarsatisfiedj=\valtrue \\
            \mergestring{\spacesymboltwo\satnotvarj}{\spacesymbol} & \mathtt{else}
        \end{array}
        \right] \\
        \merges_4 = \bigcirc_{j=1}^{\satnvariables} \left[
        \begin{array}{lr}
            \mergestring{\spacesymbol\satyesvarj}{\spacesymbol} & \mathtt{if}\,\,\satvarsatisfiedj=\valtrue \\
            \mergestring{\spacesymbol}{\satnotvarj\spacesymbol} & \mathtt{else}
        \end{array}
        \right] 
    \end{align}
    \end{subequations}
    Finally, we create a merge sequence by concatenating these lists in order:
    \begin{align}
        \merges = \merges_1 \circ \merges_2 \circ \merges_3 \circ \merges_4 \circ \merges_5
    \end{align}
    Note that we have exactly $\vocabsize=8\satnvariables$ merges in this list.
    Given this merge sequence, it is easy to verify that strings in $\dataset_1$ to $\dataset_4$ will use exactly $(6\nrepeatvar \!+\! 6\nrepeatvar' \!+\! 4\nrepeatvar'' \!+\! 4\nrepeatvar''')\,\satnvariables$ symbols after tokenised.
    Further, since at least $\minclauses$ of the \maxsatacron problem are satisfied by $\satvarsatisfiedj$, the strings in $\dataset_5$ will occupy a total length smaller or equal to $3\,\satnclauses - \minclauses$.
    This solution to the tokenisation problem thus gives us a tokeniser which will compress $\dataset$ to at most $\maxsymbols = (6\nrepeatvar \!+\! 6\nrepeatvar' \!+\! 4\nrepeatvar'' \!+\! 4\nrepeatvar''')\,\satnvariables + 3\,\satnclauses - \minclauses$.
\end{proof}

We now prove the backward direction of the iff clause in \cref{eq:bottomup_reduction_iff}.

\begin{restatable}{lemma}{bottomupnphardonlyiflemma} 
\label{lemma:bottomup_nphard_onlyif_lemma}
    If the bottom-up tokenisation instance output by \cref{reduction:max2set_to_bottomup_reduction} is satisfiable, the \maxsatacron instance reduced to it is as well. Formally:
    \begin{align}
        &\mintokmerge(\reductiontwofuncfull) \implies \maxsatfull
    \end{align}
\end{restatable}
\begin{proof}[Proof sketch]
    See a formal proof in \cref{appendix:proof_bottomup_nphard_onlyif_lemma}.
    Our proof works in five steps.
    First, we show that all satisfying solutions 
    must include merges $\merges_1$, $\merges_3$, and $\merges_5$ from \cref{eq:bottomup_merges_always}, since this is required to compress strings in $\dataset_1$ to at most $6\nrepeatvar \satnvariables$ symbols.
    Second, we show the other merges of any satisfying solution must be of the form:
    \begin{subequations}
    \begin{align}
        \mergesthreecharsone = \bigg\{
\!\!\begin{array}{c}
\mergestring{\spacesymbol\satyesvarj}{\spacesymbol}, \mergestring{\spacesymbol}{\satnotvarj\spacesymbol} \\
\mergestring{\spacesymbol}{\satyesvarj\spacesymbol}, \mergestring{\spacesymbol\satnotvarj}{\spacesymbol}
\end{array}\!\!
        \bigg\} \\
        \mergesthreecharstwo = \bigg\{\!\!
\begin{array}{c}
\mergestring{\spacesymbol}{\satyesvarj\spacesymboltwo}, \mergestring{\spacesymboltwo\satnotvarj}{\spacesymbol}\\
\mergestring{\spacesymbol\satyesvarj}{\spacesymboltwo}, \mergestring{\spacesymboltwo}{\satnotvarj\spacesymbol} 
\end{array}\!\!\bigg\}
    \end{align}
    \end{subequations}
    this is required to compress strings in $\dataset_2$ to at most $6\nrepeatvar' \satnvariables$ symbols.
    Third, we show that any satisfying solution will have at least one merge of each set $\mergesthreecharsone$ and one of each set $\mergesthreecharstwo$; this is required to compress strings in $\dataset_3$ to at most $4\nrepeatvar'' \satnvariables$ symbols.
    Fourth, we show that any satisfying solution will have---for each $j\in\{1,\satnvariables\}$---both its merges in sets $\mergesthreecharsone$ and $\mergesthreecharstwo$ containing character $\charstring{\satyesvarj}$ or containing character $\charstring{\satnotvarj}$; this is required to compress strings in $\dataset_4$ to at most $4\nrepeatvar''' \satnvariables$ symbols.
    Finally, we show that if a tokeniser compresses strings in $\dataset_5$ to $3\satnclauses - \minclauses$, then there is an assignment $\satvals$ which satisfies at least $\minclauses$ of the original \maxsatacron problem.
\end{proof}

Finally,
given both lemmas above, we can now prove that bottom-up tokenisation is NP-hard.

\begin{lemma} \label{lemma:bottomup_is_nphard}
    The bottom-up tokenisation decision problem, as in \cref{defn:bottomup_decision_problem}, is NP-hard.
\end{lemma}
\begin{proof}
    First, it is easy to see that \cref{reduction:max2set_to_bottomup_reduction} runs in polynomial time.
    Second, \maxsatacron is an NP-hard problem \cite{garey1974maxsat}.
    This lemma then follows trivially from \cref{lemma:bottomup_nphard_if_lemma,lemma:bottomup_nphard_onlyif_lemma}.
\end{proof}

\subsection{Other Definitions of Tokenisation}
\label{sec:other_tokenisation_problems}

We now expand our discussion to consider variations of the above tokenisation problems.

\paragraph{Deduped Datasets.}
Our definitions of both direct and bottom-up tokenisation allow datasets $\dataset$ to include repeated entries.
It is common, however, to deduplicate datasets in NLP---thus removing repeated entries.
A small change to both our reductions is enough to adapt it to this deduplicated dataset case: simply append each string in the repeated datasets (either $\dataset_1$ and $\dataset_2$ in \cref{reduction:max2set_to_tokenisation_reduction} or $\dataset_1$ to $\dataset_4$ in \cref{reduction:max2set_to_bottomup_reduction}) with a unique character $\{\charstring{a}_{y}\}_{y=1}^{\infty}$ and increase the target compression size $\maxsymbols$ accordingly (by $\nrepeatvar+\nrepeatvar'$ or $\nrepeatvar+\nrepeatvar'+\nrepeatvar''+\nrepeatvar'''$, respectively).
These new characters will never be included in optimal tokenisers' solutions, and thus the previous proofs hold, with the difference that each dataset will require extra symbols once compressed.\looseness=-1

\paragraph{A Single Long String.}
In the previous sections, we considered tokenisers trained on a dataset $\dataset$.
Work on compression, however, usually considers a single long string $\characters$ as its input.
It is easy to see that direct tokenisation is not an NP-complete problem if its input is a single long string; including this string in vocabulary $\vocab$ already achieves optimal compression.
Bottom-up tokenisation, however, is still NP-complete even when given a single string as input.
As before, this can be shown with a similar strategy to \cref{reduction:max2set_to_bottomup_reduction}, but where we first append each string in dataset $\dataset$ with a unique character $\{\charstring{a}_{y}\}_{y=1}^{\infty}$ and then concatenate all these strings.
As in the deduped case above, characters $\charstring{a}_y$ will never be merged by any optimal tokeniser; they will thus serve as virtual string delimiters and will not affect our proofs beyond an increase to the target compression size $\maxsymbols$.

\paragraph{A Hybrid Approach.}
Finally, the last variant we consider is a hybrid between direct and bottom-up tokenisation, where we find a merge sequence $\merges$ which---when we extract a vocabulary from it as $\vocab = \alphabet\cup\{\subword_1 \circ \subword_2 \mid \mergestring{\subword_1}{\subword_2} \in \merges\}$---optimally compresses a dataset $\dataset$ using the direct tokenisation function in \cref{eq:direct_tok_function}.
We can easily prove the NP-hardness of this tokenisation variant by relying on \cref{reduction:max2set_to_bottomup_reduction}; as our proof in \cref{lemma:bottomup_is_nphard} did not make use of the order of merges in $\merges$, only of the subwords composed by it, this lemma's proof strategy can be similarly applied to this hybrid variant. 

\section{Tokenisation's Connection to Compression}

The variants of tokenisation that we consider here---with compression as their objective function---are closely related to the field of dictionary compression.
In both fields, we wish to reduce the size of an input ($\characters$ or $\dataset$) by exploiting repetitive elements.
In fact, the most popular tokenisation algorithm to date, BPE, was originally proposed as a compression algorithm \citep{gage1994new} and has only somewhat recently been ported into NLP to find tokenisers \citep[by][]{sennrich-etal-2016-neural}.

Not surprisingly, prior work has also considered, from a theoretical perspective, the compression tokenisers achieve.
\citet{zouhar-etal-2023-formal}, for instance, analyse bottom-up tokenisation and prove an approximation bound on the compression achieved by the tokenisers found using BPE.
More recently, \citet{kozma2024theoretical} also analyses bottom-up tokenisation, proving a tighter bound on this compression achieved by BPE.

A set of popular dictionary compression methods, \defn{straight-line programs} \citep[SLP;][]{kieffer2000slp, charikar2005slp}, can be used to illustrate the similarities and differences between tokenisers and compressors.\footnote{
See \citet{lohrey2012slpsurvey} for an overview of straight-line programs, and \citet{kempa2018stringattractors, kociumaka2023deltameasure} for a more detailed overview of compression in general. 
}
Given a string $\characters$, an SLP describes a grammar from which $\characters$ can be uniquely derived.
Formally, an SLP is a set of rules of form $\subwordstring{X} \to \charstring{a}$ or $\subwordstring{X} \to \subwordstring{A} \subwordstring{B}$, where $\subwordstring{X}, \subwordstring{A}, \subwordstring{B}$ are called nonterminals and $\charstring{a}$ is a terminal.\footnote{
Although not originally defined that way, SLP's grammars are typically assumed to be in Chomsky normal form (CNF), for simplicity. This does not make a big difference for compression, but will be important for our purposes.
}
Starting from a special nonterminal $\subwordstring{S}$, applying these rules exhaustively---until only terminals are left---produces exactly the desired string $\characters$.
Notably, given a string $\characters$, it is NP-complete to find the smallest SLP which generates it \cite{charikar2005slp}.

On the one hand, SLPs in Chomsky normal form are closely linked to bottom-up tokenisation; each of its rules expands to two nonterminals, and thus corresponds to a merge.
However, while SLPs must find the minimum number of merges (or rules) to fully compress a string into a single symbol, bottom-up tokenisers must maximally compress the string given a fixed number of merges.
On the other hand, SLPs which are not in Chomsky normal form are closely linked to direct tokenisation.
In this case, a direct tokeniser could be converted into an SLP with depth two; this grammar has a start rule $S \to \subwords$, and a rule from each subword to its characters $\subword \to \characters$.
Again, while SLPs must find a minimal grammar representing the string, direct tokenisers must only minimise the size of rule $S \to \subwords$ given a fixed number of rules $\subword \to \characters$.

The paragraph above highlights two important differences between tokenisers and compressors.
First, tokenisers aim to reduce only the size of the resulting tokenised text (i.e., $|\subwords|$), whereas compressors also consider the size of the compression information (e.g., considering the size required to store $\vocab$, which would be $\sum_{\subword \in \vocab} |\detokenise(\subword)|$).
This is because tokenisers must create shorter inputs for NLP algorithms, while compressors must make information compact.
Second, tokenisers and compressors have different optimisation parameters.
Compression algorithms always compress a string to the best extent possible (e.g., for SLPs, until a single nonterminal is reached).
Instead, tokenisation algorithms are given a maximum vocabulary size (i.e., $\vocabsize$) and find tokenisers which only compress their input as much as possible until this limit is reached.

\vspace{-2pt}
\section{Conclusion}
\vspace{-2pt}

In this work, we proved the NP-completeness of two variants of tokenisation.
These results underline that finding optimal tokenisers most likely will remain a difficult quest and that research should focus on approximate algorithms instead.
Regarding those, there is potential both in improving the analysis of currently used algorithms, such as BPE, as well as in designing other, more involved algorithms.
Towards the latter, a more detailed look at the connections to compression algorithms might be fruitful.
While we investigated the complexity of two forms of tokenisation, similar results for other variants (e.g., with other objective functions) remain open; this would be exciting future work.

\section*{Acknowledgments}

We would like to thank Sotiris Anagnostidis and Clara Meister for their feedback on earlier versions of this manuscript.
We would also like to thank Amit Moryossef for discussions about the NP-hardness proof we present here.
Finally, we also thank the VMI of ETH Z\"urich for organising a scientific workshop that allowed the authors to meet and exchange ideas leading to this paper.\looseness=-1

\bibliography{custom}

\onecolumn
\appendix

\section{Proof of \texorpdfstring{\Cref{lemma:direct_nphard_if_lemma}}{Lemma}}
\label{appendix:proof_direct_nphard_if_lemma}

\directnphardiflemma*
\begin{proof}
    First, note that if $\maxsatfull$, then we have that \cref{eq:max2sat_satisfiable} holds:
    \begin{align}
        \minclauses \leq 
        &\max_{\satvals \in \{\valfalse,\valtrue\}^{J}} \sum_{i=1}^{I} \one\{\satliteral_i^1 \lor \satliteral_i^2\}
    \end{align}
    Now, without loss of generality, let a satisfying solution have values $\satvarsatisfiedj$.
    In this case, for each variable $\satvar_j$, we construct token \token{\spacesymbol\satyesvarj\spacesymbol} if $\satvarsatisfiedj$ is true, or
    \token{\spacesymbol\satnotvarj\spacesymbol} if $\satvarsatisfiedj$ is false.
    This gives us a total of $J$ new tokens, so satisfies the $|\vocab| = |\alphabet| + \vocabsize$ condition.
    Now we just need to count the symbols output by this solution to see if \cref{eq:direct_satisfiable} is satisfied, since any given tokenisation $\tokenise(\cdot, \vocab)$ will provide an upper bound on the optimal tokenisation in terms of compression:
    \begin{align}
        \sum_{\characters \in \dataset} |\directtokenfull|
        \geq &\min_{\vocab' \subset \alphabet^*} \sum_{\characters \in \dataset} |\directtoken[\vocab'](\characters)| \\
        & \mathrm{s.t.}\,\,|\vocab'| = |\alphabet| + \vocabsize \nonumber
    \end{align}
    For each pair of strings $\charstring{\spacesymbol\satyesvarj\spacesymbol}$ and $\charstring{\spacesymbol\satnotvarj\spacesymbol}$ in $\dataset_1$, one is compressed into a single subword 
    while the other is kept as originally---using 3 symbols.
    We thus have that the strings in $\dataset_1$ will occupy a total of $(1 + 3)J$ characters, and:
    \begin{align} \label{eq:nphard_if_proof_dataset_1}
        \sum_{\characters \in (\bigcup_{\_ = 1}^{\nrepeatvar} \dataset_1)} |\directtokenfull| = 4 \nrepeatvar J
    \end{align}
    Similarly, for each string in $\dataset_2$ of form $\charstring{\spacesymbol\satyesvarj\spacesymbol\satnotvarj\spacesymbol}$, we have that either token $\subwordstring{\spacesymbol\satyesvarj\spacesymbol}$ or $\subwordstring{\spacesymbol\satnotvarj\spacesymbol}$ exists.
    So each of these strings is compressed from 5 into 3 symbols.
    We thus have:
    \begin{align} \label{eq:nphard_if_proof_dataset_2}
        \sum_{\characters \in (\bigcup_{\_ = 1}^{\nrepeatvar'} \dataset_2)} |\directtokenfull| = 3 \nrepeatvar' J
    \end{align}
    Finally, we have strings in $\dataset_3$ of form $\charstring{\spacesymbol\satliteral_i^1\spacesymbol\satliteral_i^2\spacesymbol}$.
    These strings will be compressed into 3 symbols if 
    $\subwordstring{\spacesymbol\satliteral_i^1\spacesymbol}$ or $\subwordstring{\spacesymbol\satliteral_i^2\spacesymbol}$ (or both) exist, and kept with 5 symbols otherwise.
    We thus have:
    \begin{subequations} \label{eq:nphard_if_proof_dataset_3}
    \begin{align}
        \sum_{\characters \in \dataset_3} |\directtokenfull| 
        &= \sum_{\characters \in \dataset_3} \Bigg(5 - 2\, \one\Bigg\{
        \!\!\!
        \begin{array}{c}
            \token{\spacesymbol\satliteral_i^1\spacesymbol} \in \vocab  \\
            \mathrm{or} \\
            \token{\spacesymbol\satliteral_i^2\spacesymbol} \in \vocab
        \end{array}
        \!\!
        \Bigg\} \Bigg) \\
        & = 5I - 2\sum_{\characters \in \dataset_3} \one\Bigg\{
        \!\!\!
        \begin{array}{c}
            \token{\spacesymbol\satliteral_i^1\spacesymbol} \in \vocab  \\
            \mathrm{or} \\
            \token{\spacesymbol\satliteral_i^2\spacesymbol} \in \vocab
        \end{array}
        \!\!
        \Bigg\}  \\
        & = 5I - 2 \sum_{i=1}^{I} 
            \one\{\satliteral_i^1  \lor \satliteral_i^2 \} \\
        & \leq 5I - 2 \minclauses
    \end{align}
    \end{subequations}
    where, by construction, we have that a subword $\subwordstring{\spacesymbol\satliteral_i\spacesymbol} \in \vocab$ if and only if its associated variable ($\satval_j$ or $\neg \satval_j$) is true.
    Summing together the lengths in \cref{eq:nphard_if_proof_dataset_1,eq:nphard_if_proof_dataset_2,eq:nphard_if_proof_dataset_3}, we get that 
    \begin{align}
        \sum_{\characters \in \dataset} |\directtokenfull| \leq \maxsymbols = (4\nrepeatvar + 3\nrepeatvar')\,J + 5\,I - 2\minclauses
    \end{align}
    which concludes the proof.
\end{proof}

\section{Proof of \texorpdfstring{\Cref{lemma:direct_nphard_onlyif_lemma}}{Lemma}}
\label{appendix:proof_direct_nphard_onlyif_lemma}

\directnphardonlyiflemma*
\begin{proof}
    First, note that the dataset $\dataset$ output by \cref{reduction:max2set_to_tokenisation_reduction} has a total of characters:
    \begin{align}
        \sum_{\characters \in \dataset} |\characters| = (6\nrepeatvar + 5\nrepeatvar') \satnvariables + 5\satnclauses
    \end{align}
    Further, let:
    \begin{align}
    	\tokeniselength \defeq \sum_{\characters \in \dataset} |\directtokenfull|,
    	\qquad\qquad
    	\vocabvariables = \bigcup_{j=1}^{J}
        \{\subwordstring{\spacesymbol\satyesvarj\spacesymbol}, \subwordstring{\spacesymbol\satnotvarj\spacesymbol}\}
    \end{align}
    The maximum number of symbols in this dataset after compression is set to $\maxsymbols = (4\nrepeatvar + 3\nrepeatvar')\,\satnvariables + 5\,\satnclauses - 2\minclauses$.
    This means that, to satisfy this objective, there must exist a vocabulary whose tokeniser compresses the text by at least $(2\nrepeatvar + 2\nrepeatvar')\,\satnvariables + 2\minclauses$ symbols.
    We now prove this lemma in three steps: \circled{1} we show that any solution which compresses the text by at least $2\nrepeatvar \satnvariables$ symbols must only have subwords of the form \subwordstring{$\spacesymbol\satyesvarj\spacesymbol$} or \subwordstring{$\spacesymbol\satnotvarj\spacesymbol$};
    \circled{2} we show that any solution which compresses the text by at least $(2\nrepeatvar  + 2\nrepeatvar')\satnvariables$ symbols must only have either subword \subwordstring{$\spacesymbol\satyesvarj\spacesymbol$} or \subwordstring{$\spacesymbol\satnotvarj\spacesymbol$} for any variable $\satvar_j$;
    \circled{3} we show that any solution which compresses the text by at least $(2\nrepeatvar  + 2\nrepeatvar')\satnvariables + 2\minclauses$ symbols must be produced by a \maxsatacron instance which has at least $\minclauses$ clauses that are simultaneously satisfiable.
\end{proof}

\begin{mylemmastep} \textnormal{(Step \circled{1}).}
Any solution which compresses the text by at least $2\nrepeatvar \satnvariables$ symbols must only have subwords of the form \subwordstring{$\spacesymbol\satyesvarj\spacesymbol$} or \subwordstring{$\spacesymbol\satnotvarj\spacesymbol$}, i.e.,: 
\begin{align}
    &\bigg(
    \tokeniselength 
    \leq \underbrace{(4\nrepeatvar + 5\nrepeatvar') \satnvariables + 5\satnclauses}_{\sum_{\characters \in \dataset} |\characters| - 2\nrepeatvar \satnvariables}\bigg) 
    \implies 
    \vocab \subset \vocabvariables
\end{align}
\end{mylemmastep}
\begin{proof}
    First, given a solution with $\vocab \subset \vocabvariables$, each subword $\subword \in \vocab$ will replace at least $\nrepeatvar$ strings in $\dataset_1$---i.e., with form $\charstring{\spacesymbol\satyesvarj\spacesymbol}$ or $\charstring{\spacesymbol\satnotvarj\spacesymbol}$---for a single subword, thus saving $2\nrepeatvar$ characters.
    Since we have $|\vocab| = \vocabsize = \satnvariables$ tokens, we save exactly $2\nrepeatvar \satnvariables$ symbols: 
    \begin{align}
        \vocab \subset \vocabvariables
        \implies
        \bigg(\tokeniselengthfun(\dataset_1, \vocab')
        = \underbrace{4\nrepeatvar \satnvariables}_{\sum_{\characters \in \dataset_1} |\characters| - 2\nrepeatvar \satnvariables}\bigg) 
    \end{align}
    Note now that any solution $\vocab'$ for which $\vocab' \not\subset \vocabvariables$ has at least one subword which is not of the form $\subwordstring{\spacesymbol\satyesvarj\spacesymbol}$ or $\subwordstring{\spacesymbol\satnotvarj\spacesymbol}$; 
    this subword $\subword \notin \vocabvariables$ will thus not compress strings in $\dataset_1$ by $2\nrepeatvar$ symbols, but by at most $\nrepeatvar$:
    \begin{align}
        \vocab' \not\subset \vocabvariables
        \implies
        \bigg(\tokeniselengthfun(\dataset_1, \vocab')
        \geq \underbrace{4\nrepeatvar (\satnvariables - 1) + 5\nrepeatvar}_{\sum_{\characters \in \dataset_1} |\characters| - 2\nrepeatvar \satnvariables + \nrepeatvar}\bigg) 
    \end{align}
    Even if this new subword were able to fully compress strings in $\dataset_2$ and $\dataset_3$ to a single symbol each, it would reach a compression of at most $4\nrepeatvar'\satnvariables + 4\satnclauses$.
    Since by design $\nrepeatvar = 4\nrepeatvar' \satnvariables + 4\satnclauses + 1$,
    we get that:
    \begin{align}
        \vocab' \not\subset \vocabvariables
        \implies
        \bigg(\tokeniselengthfun(\dataset, \vocab')
        \geq 
        4\nrepeatvar \satnvariables + \nrepeatvar + 
         \nrepeatvar' \satnvariables + \satnclauses
         > (4\nrepeatvar + 5\nrepeatvar') \satnvariables + 5\satnclauses
         \bigg) 
    \end{align}
    which concludes this step of the proof.
\end{proof}
\begin{mylemmastep} \textnormal{(Step \circled{2}).}
Any solution which compresses the text by at least $(2\nrepeatvar  + 2\nrepeatvar')\satnvariables$ symbols must only have either subword \subwordstring{$\spacesymbol\satyesvarj\spacesymbol$} or \subwordstring{$\spacesymbol\satnotvarj\spacesymbol$} for any variable $\satvar_j$, i.e.,:
\begin{align}
    &\bigg(
    \tokeniselength 
    \leq \underbrace{(4\nrepeatvar + 3\nrepeatvar') \satnvariables + 5\satnclauses}_{\sum_{\characters \in \dataset} |\characters| - (2\nrepeatvar  + 2\nrepeatvar')\satnvariables}\bigg) 
    \implies 
    \forall_{j \in \{1, ..., J\}}\,\,
    |\vocab \cap 
    \{\subwordstring{\spacesymbol\satyesvarj\spacesymbol}, \subwordstring{\spacesymbol\satnotvarj\spacesymbol}\}| = 1
\end{align}
\end{mylemmastep}
\begin{proof}
    In this step of the proof, we show that satisfying solutions must have one and only one of subwords $\subwordstring{\spacesymbol\satyesvarj\spacesymbol}$ and  $\subwordstring{\spacesymbol\satnotvarj\spacesymbol}$ for any variable $\satvar_j$.
    As before, it's easy to see that a solution of the form described achieves at least $(2\nrepeatvar  + 2\nrepeatvar')\satnvariables$ symbol compression.
    Each subword of form $\subwordstring{\spacesymbol\satyesvarj\spacesymbol}$ or  $\subwordstring{\spacesymbol\satnotvarj\spacesymbol}$ saves exactly $2\nrepeatvar$ characters in the strings in $\dataset_1$.
    Further, because we always have either subword $\subwordstring{\spacesymbol\satyesvarj\spacesymbol}$ or  $\subwordstring{\spacesymbol\satnotvarj\spacesymbol}$ for each value of $j$, we also get $2\nrepeatvar'$ compression in the strings in $\dataset_2$:
    \begin{align} \label{eq:direct_compression_datasets12}
        &\forall_{j \in \{1, ..., J\}}\,\,
        |\vocab \cap 
        \{\subwordstring{\spacesymbol\satyesvarj\spacesymbol}, \subwordstring{\spacesymbol\satnotvarj\spacesymbol}\}| = 1 \\
        &\qquad\qquad\qquad
        \implies 
        \bigg(
        \tokeniselengthfun(\dataset_1, \vocab)
        = \underbrace{4\nrepeatvar\satnvariables}_{\sum_{\characters \in \dataset_1} |\characters| - 2\nrepeatvar\satnvariables}\bigg) 
        \,\mathtt{and}\,
        \bigg(
        \tokeniselengthfun(\dataset_2, \vocab)
        = \underbrace{3\nrepeatvar'\satnvariables}_{\sum_{\characters \in \dataset_2} |\characters| - 2\nrepeatvar'\satnvariables}\bigg) 
        \nonumber
    \end{align}
    Now note that this is not true if both $\subwordstring{\spacesymbol\satyesvarj\spacesymbol}$ and  $\subwordstring{\spacesymbol\satnotvarj\spacesymbol}$ exist for a single $j$; in this case, only one of the tokens can be applied to $\subwordstring{\spacesymbol\satyesvarj\spacesymbol\satnotvarj\spacesymbol}$, and thus both tokens together lead to a benefit of $2$ instead of $4$.
    If both $\subwordstring{\spacesymbol\satyesvarj\spacesymbol}$ and  $\subwordstring{\spacesymbol\satnotvarj\spacesymbol}$ exist for any token $\satvar_j$, this implies that neither of $\subwordstring{\spacesymbol\satyesvarjprime\spacesymbol}$ and $\subwordstring{\spacesymbol\satnotvarjprime\spacesymbol}$ exists for some other $\satvar_{j'}$, resulting in an uncompressed string.
    Then, we get at most a compression of $2\nrepeatvar \satnvariables  + 2\nrepeatvar'(\satnvariables-1) + 4\satnclauses$:
    \begin{align}
        \exists_{j \in \{1, ..., J\}}\,\,
        |\vocab' \cap 
        \{\subwordstring{\spacesymbol\satyesvarj\spacesymbol}, \subwordstring{\spacesymbol\satnotvarj\spacesymbol}\}| \neq 1
        \implies
        \bigg(\tokeniselengthfun(\dataset, \vocab')
        \geq 
        \underbrace{
        (4\nrepeatvar + 3\nrepeatvar') \satnvariables + \nrepeatvar' + \satnclauses
         }_{\sum_{\characters \in \dataset} |\characters| - (2\nrepeatvar + 2\nrepeatvar)'\satnvariables + \nrepeatvar' - 4\satnclauses}
         \bigg) 
    \end{align}
    By construction 
    $\nrepeatvar' = 4\satnclauses+1$,
    which leads to:
    \begin{align}
        \exists_{j \in \{1, ..., J\}}\,\,
        |\vocab' \cap 
        \{\subwordstring{\spacesymbol\satyesvarj\spacesymbol}, \subwordstring{\spacesymbol\satnotvarj\spacesymbol}\}| \neq 1
        \implies
        \bigg(\tokeniselengthfun(\dataset, \vocab')
        > (4\nrepeatvar + 3\nrepeatvar') \satnvariables + 5\satnclauses
         \bigg) 
    \end{align}
    This concludes the proof.
\end{proof}
\begin{mylemmastep} \textnormal{(Step \circled{3}).}
Any tokenisation problem with a solution which compresses the text by at least $(2\nrepeatvar  + 2\nrepeatvar')\satnvariables + 2\minclauses$ symbols must be produced by a \maxsatacron problem with at least $\minclauses$ satisfied clauses, i.e.,:
\begin{align}
    &\bigg(
    \tokeniselength 
    \leq \underbrace{(4\nrepeatvar + 3\nrepeatvar') \satnvariables + 5\satnclauses - 2\minclauses}_{\sum_{\characters \in \dataset} |\characters| - (2\nrepeatvar  + 2\nrepeatvar')\satnvariables + 2\minclauses}\bigg) 
    \implies 
    \maxsatfull
\end{align}
\end{mylemmastep}
\begin{proof}
    Finally, we now know any solution with this compression must have---for any variable $\satvar_j$---either subword $\subwordstring{\spacesymbol\satyesvarj\spacesymbol}$ or $\subwordstring{\spacesymbol\satnotvarj\spacesymbol}$.
    We can thus create a bijection $\bijectionvocabsat$ between the set of possible vocabularies respecting this condition, and the set of $\valtrue$/$\valfalse$ assignments to SAT variables $\satvals$:
    \begin{align}
        \bijectionvocabsat(\vocab) = \bigg\{
        \begin{array}{lr}
            \valtrue & \mathtt{if}\,\subwordstring{\spacesymbol\satyesvarj\spacesymbol}\in\vocab \\
            \valfalse & \mathtt{if}\,\subwordstring{\spacesymbol\satnotvarj\spacesymbol}\in\vocab
        \end{array}
        \bigg\}_{j=1}^{\satnvariables}
    \end{align}
    Further, note that vocabularies of this form (as shown in \cref{eq:direct_compression_datasets12}) lead to exactly $(2\nrepeatvar  + 2\nrepeatvar')\satnvariables$ symbols being compressed in $\dataset_1$ and $\dataset_2$.
    To achieve the target compression, a solution must thus compress $\dataset_3$ by at least $2\minclauses$ symbols.
    Now note that for any string $\charstring{\spacesymbol\satliteral_i^1\spacesymbol\satliteral_i^2\spacesymbol}$ in $\dataset_3$ we have three compression options: $\charstring{\spacesymbol\satliteral_i^1\spacesymbol}$ will be compressed, saving 2 symbols; 
    $\charstring{\spacesymbol\satliteral_i^2\spacesymbol}$ will be compressed, also saving 2 symbols; or nothing will be compressed.
    More specifically,
    $\charstring{\spacesymbol\satliteral_i^1\spacesymbol}$ can be compressed
    if $\satliteral_i^1$ represents $\satvar_j$ and subword $\subwordstring{\spacesymbol\satyesvarj\spacesymbol}$ exists,
    or if $\satliteral_i^1$ represents $\neg\satvar_j$ and subword $\subwordstring{\spacesymbol\satnotvarj\spacesymbol}$ exists; the same is true for $\charstring{\spacesymbol\satliteral_i^2\spacesymbol}$.
    They cannot both be compressed, however, as there is only one symbol $\charstring{\spacesymbol}$ between the literals.
    We thus get a compression of $2$ symbols for each of these strings if at least one of its literals has an associated subword in $\vocab$.
    Note thus that whenever a string $\charstring{\spacesymbol\satliteral_i^1\spacesymbol}$ is compressed by 2 symbols using vocabulary $\vocab$, the \maxsatacron disjunction $\satliteral_i^1 \lor \satliteral_i^2$ will also be satisfied by assignment $\satvals = \bijectionvocabsat(\vocab)$;
    similarly, whenever this string suffers no compression (i.e., having a compression of zero), the \maxsatacron disjunction will not be satisfied.
    As our condition assumes a compression of at least $2\minclauses$ symbols, we know that we have at least $\minclauses$ strings for which a literal has an associated subword.
    We can thus write:
    \begin{subequations}
    \begin{align}
         2\minclauses \leq
        &\max_{\vocab \subset \alphabet^*} \sum_{\characters \in \dataset_3} |\characters| - |\directtokenfull|  \\
        & \mathrm{s.t.}\,\,
        |\vocab| = \satnvariables \,\,\mathtt{and}\,\, \forall_{j \in \{1, ..., \satnvariables\}}\,\,
        |\vocab \cap 
        \{\subwordstring{\spacesymbol\satyesvarj\spacesymbol}, \subwordstring{\spacesymbol\satnotvarj\spacesymbol}\}| = 1 \nonumber \\
        =
        &\max_{\vocab \subset \alphabet^*} \,\,\,\,\, \sum_{\charstring{\spacesymbol\satliteral_i^1\spacesymbol\satliteral_i^2\spacesymbol} \,\in\, \dataset_3} 2\,
        \one\Bigg\{
        \begin{array}{c}
             \subwordstring{\spacesymbol\satliteral_i^1\spacesymbol} \in \vocab  \\
               \mathrm{or} \\
             \subwordstring{\spacesymbol\satliteral_i^2\spacesymbol} \in \vocab
        \end{array} \Bigg\}  \\
        & \mathrm{s.t.}\,\,
        |\vocab| = \satnvariables \,\,\mathtt{and}\,\, \forall_{j \in \{1, ..., \satnvariables\}}\,\,
        |\vocab \cap 
        \{\subwordstring{\spacesymbol\satyesvarj\spacesymbol}, \subwordstring{\spacesymbol\satnotvarj\spacesymbol}\}| = 1 \nonumber \\
        =&\max_{\satvals \in \{0,1\}^{\satnvariables}} \sum_{i=1}^{\satnclauses} 2\one\{\satliteral_i^1 \lor \satliteral_i^2\} \\
        &\qquad\qquad \implies\,\, \maxsatfull
    \end{align}
    \end{subequations}
    We thus know that, if a satisfying tokenisation solution exists, then the associated \maxsatacron problem will also be satisfiable.
    This concludes the proof.
\end{proof}

\section{Proof of \texorpdfstring{\Cref{lemma:bottomup_nphard_if_lemma}}{Lemma}}
\label{appendix:proof_bottomup_nphard_if_lemma}

\bottomupnphardiflemma*
\begin{proof}
    Our proof starts by first defining the three following lists of merges, which will be included in any satisfying solution to the tokenisation problem:
    \begin{align}
        \merges_1 = \BigCirc_{j=1}^{\satnvariables} \big[\mergestring{\spacesymboltwo}{\satnotvarj}, \mergestring{\satyesvarj}{\spacesymboltwo}\big], \quad
        \merges_3 = \BigCirc_{j=1}^{\satnvariables} \big[\mergestring{\satnotvarj}{\spacesymbol}, \mergestring{\spacesymbol}{\satyesvarj}\big], \quad
        \merges_5 = \BigCirc_{j=1}^{\satnvariables} \big[\mergestring{\spacesymbol}{\satnotvarj}, \mergestring{\satyesvarj}{\spacesymbol}\big]
    \end{align}
    Now, without loss of generality, let a satisfying solution to \maxsatacron have values $\satvarsatisfiedj$.
    We then construct two other lists of merges, which depend on this \maxsatacron solution:
    \begin{align}
        \merges_2 = \BigCirc_{j=1}^{\satnvariables} \left[
        \begin{array}{lr}
            \mergestring{\spacesymbol}{\satyesvarj\spacesymboltwo} & \mathtt{if}\,\,\satvarsatisfiedj=\valtrue \\
            \mergestring{\spacesymboltwo\satnotvarj}{\spacesymbol} & \mathtt{else}
        \end{array}
        \right], \qquad
        \merges_4 = \BigCirc_{j=1}^{\satnvariables} \left[
        \begin{array}{lr}
            \mergestring{\spacesymbol\satyesvarj}{\spacesymbol} & \mathtt{if}\,\,\satvarsatisfiedj=\valtrue \\
            \mergestring{\spacesymbol}{\satnotvarj\spacesymbol} & \mathtt{else}
        \end{array}
        \right] 
    \end{align}
    in words, we create merges $\mergestring{\spacesymbol}{\satyesvarj\spacesymboltwo}$ and $\mergestring{\spacesymbol\satyesvarj}{\spacesymbol}$ if  $\satvarsatisfiedj$ is true, or
    $\mergestring{\spacesymboltwo\satnotvarj}{\spacesymbol}$ and 
$\mergestring{\spacesymbol}{\satnotvarj\spacesymbol}$ if $\satvarsatisfiedj$ is false.
    We then create a merge sequence by concatenating these lists in order:
    \begin{align}
        \merges = \merges_1 \circ \merges_2 \circ \merges_3 \circ \merges_4 \circ \merges_5
    \end{align}
    This gives us a total of $|\merges| = \vocabsize = 8\satnvariables$ merges.
    Now we just need to count the symbols output by this solution to see if \cref{eq:bottomup_tokenisation_satisfiable} is satisfied, since any given tokenisation $\bottomuptoken[\merges]$ will provide an upper bound on the optimal tokenisation in terms of compression:
    \begin{align}
        \sum_{\characters \in \dataset} |\bottomuptokenfull|
        \geq &\min_{\merges' \subset \mergeset^*} \sum_{\characters \in \dataset} |\bottomuptoken[\merges'](\characters)| \\
        & \mathrm{s.t.}\,\,|\merges'| = \vocabsize \nonumber
    \end{align}
    By applying the merges $\merges$, each string in $\dataset_1$ will be compressed into a single subword; note that $\merges_2$ and $\merges_4$ will have no effect on these strings.
    This is easy to see by applying merges sequentially to these strings, as displayed in \cref{tab:bottomup_dataset1_nphard_if}.
    leading to:
    \begin{align} \label{eq:bottomup_nphard_if_proof_dataset_1}
        \sum_{\characters \in (\bigcup_{\_ = 1}^{\nrepeatvar} \dataset_1)} |\bottomuptokenfull| = 6 \nrepeatvar J
    \end{align}
    \begin{table}[t]
        \centering
        \resizebox{\textwidth}{!}{%
        \begin{tabular}{ccccc}
            \toprule
            $\characters$ 
            & $\bottomuptoken[\merges_1](\characters)$ 
            & $\bottomuptoken[\merges_1 \circ \merges_2 \circ \merges_3](\characters )$ 
            & $\bottomuptoken[\merges_1 \circ \merges_2 \circ \merges_3 \circ \merges_4 \circ \merges_5](\characters)$
            & $|\bottomuptokenfull|$ \\
            \midrule
$\subwordstring{\langle\spacesymbol,\satyesvarj\rangle}$ & $\cdot$ & $\subwordstring{\langle\spacesymbol\satyesvarj\rangle}$ & $\cdot$& 1 \\
$\subwordstring{\langle\satnotvarj,\spacesymbol\rangle}$ & $\cdot$ & $\subwordstring{\langle\satnotvarj\spacesymbol\rangle}$ & $\cdot$& 1 \\
$\subwordstring{\langle\satyesvarj,\spacesymbol\rangle}$ & $\cdot$ & $\cdot$ & $\subwordstring{\langle\satyesvarj\spacesymbol\rangle}$ & 1 \\
$\subwordstring{\langle\spacesymbol,\satnotvarj\rangle}$ & $\cdot$ & $\cdot$ & $\subwordstring{\langle\spacesymbol\satnotvarj\rangle}$ & 1 \\
$\subwordstring{\langle\satyesvarj,\spacesymboltwo\rangle}$ & $\subwordstring{\langle\satyesvarj\spacesymboltwo\rangle}$ & $\cdot$ & $\cdot$ & 1 \\
$\subwordstring{\langle\spacesymboltwo,\satnotvarj\rangle}$ & $\subwordstring{\langle\spacesymboltwo\satnotvarj\rangle}$ & $\cdot$ & $\cdot$ & 1 \\
\bottomrule
        \end{tabular}
        }
        \caption{Example of applying $\merges$ in $\dataset_1$ of bottom-up tokenisation problem obtained from \cref{reduction:max2set_to_bottomup_reduction}. The dot symbol $\cdot$ denotes the string not changing under the given merge.}
        \label{tab:bottomup_dataset1_nphard_if}
    \end{table}

    \begin{table}[t]
        \centering
        \resizebox{\textwidth}{!}{%
        \begin{tabular}{ccccccccc}
            \toprule
            $\characters$ 
            & $\bottomuptoken[\merges_1](\characters)$ 
            & \multicolumn{2}{c}{$\bottomuptoken[\merges_1 \circ \merges_2](\characters)$} 
            & $\bottomuptoken[\merges_1 \circ \merges_2 \circ \merges_3](\characters)$ 
            & \multicolumn{2}{c}{$\bottomuptoken[\merges_1 \circ \merges_2 \circ \merges_3 \circ \merges_4](\characters)$}
            & \multicolumn{2}{c}{$|\bottomuptokenfull|$} \\
            \cmidrule(lr){3-4}
            \cmidrule(lr){6-7}
            \cmidrule(lr){8-9}
            && $\satvarsatisfiedj = \valtrue$ &$\satvarsatisfiedj = \valfalse$ 
            && $\satvarsatisfiedj = \valtrue$ &$\satvarsatisfiedj = \valfalse$ 
            & $\satvarsatisfiedj = \valtrue$ &$\satvarsatisfiedj = \valfalse$ 
            \\
            \midrule
$\subwordstring{\langle\spacesymbol,\satyesvarj,\spacesymbol\rangle}$ & $\cdot$ & $\cdot$ & $\cdot$ & $\subwordstring{\langle\spacesymbol\satyesvarj,\spacesymbol\rangle}$ & 
$\subwordstring{\langle\spacesymbol\satyesvarj\spacesymbol\rangle}$ &
$\subwordstring{\langle\spacesymbol,\satyesvarj\spacesymbol\rangle}$ &
1 & 2 \\
$\subwordstring{\langle\spacesymbol,\satnotvarj,\spacesymbol\rangle}$ & $\cdot$ & $\cdot$ & $\cdot$ & $\subwordstring{\langle\spacesymbol,\satnotvarj\spacesymbol\rangle}$ & 
$\subwordstring{\langle\spacesymbol,\satnotvarj\spacesymbol\rangle}$ &
$\subwordstring{\langle\spacesymbol\satnotvarj\spacesymbol\rangle}$ &
2 & 1 \\
$\subwordstring{\langle\spacesymbol,\satyesvarj,\spacesymboltwo\rangle}$ & 
$\subwordstring{\langle\spacesymbol,\satyesvarj\spacesymboltwo\rangle}$ &
$\subwordstring{\langle\spacesymbol\satyesvarj\spacesymboltwo\rangle}$ &
$\subwordstring{\langle\spacesymbol,\satyesvarj\spacesymboltwo\rangle}$ &
$\cdot$ & $\cdot$ & $\cdot$ & 
1 & 2 \\
$\subwordstring{\langle\spacesymboltwo,\satnotvarj,\spacesymbol\rangle}$ & 
$\subwordstring{\langle\spacesymboltwo\satnotvarj,\spacesymbol\rangle}$ &
$\subwordstring{\langle\spacesymboltwo,\satnotvarj\spacesymbol\rangle}$ &
$\subwordstring{\langle\spacesymboltwo\satnotvarj\spacesymbol\rangle}$ &
$\cdot$ & $\cdot$  & $\cdot$ & 
2 & 1 \\
\bottomrule
        \end{tabular}
        }
        \caption{Example of applying $\merges$ in $\dataset_2$ of bottom-up tokenisation problem obtained from \cref{reduction:max2set_to_bottomup_reduction}. The dot symbol $\cdot$ denotes the string not changing under the given merge.}
        \label{tab:bottomup_dataset2_nphard_if}
    \end{table}

    For each pair of strings $\charstring{\spacesymbol\satyesvarj\spacesymbol}$ and $\charstring{\spacesymbol\satnotvarj\spacesymbol}$ in $\dataset_2$, one is compressed into a single subword while the other is only compressed to two subwords---the one with $\charstring{\satyesvarj}$ is compressed to a single symbol if $\satvarsatisfiedj = \valtrue$ and the one with $\charstring{\satnotvarj}$ otherwise.
    The same is true for each pair of strings $\charstring{\spacesymbol\satyesvarj\spacesymboltwo}$ and $\charstring{\spacesymboltwo\satnotvarj\spacesymbol}$, also in $\dataset_2$.
    This is displayed in \cref{tab:bottomup_dataset2_nphard_if}.
    We thus have that, for each variable $\satvar_j$, the strings in $\dataset_2$ will occupy a total of $(1 + 2 + 1 + 2)J$ characters, and:
    \begin{align} \label{eq:bottomup_nphard_if_proof_dataset_2}
        \sum_{\characters \in (\bigcup_{\_ = 1}^{\nrepeatvar} \dataset_1)} |\bottomuptokenfull| = 6 \nrepeatvar' J
    \end{align}
    Similarly, each string in $\dataset_3$ and $\dataset_4$ will be compressed into only 2 symbols after this tokeniser is applied to it.
    We thus have:
    \begin{align} \label{eq:bottomup_nphard_if_proof_dataset_34}
        \sum_{\characters \in (\bigcup_{\_ = 1}^{\nrepeatvar''} \dataset_3)} |\bottomuptokenfull| = 4 \nrepeatvar'' \satnvariables,
        \qquad
        \sum_{\characters \in (\bigcup_{\_ = 1}^{\nrepeatvar'''} \dataset_4)} |\bottomuptokenfull| = 4 \nrepeatvar''' \satnvariables
    \end{align}
    Finally, we have the strings in $\dataset_5$.
    These strings are constructed such that they will be compressed into 2 symbols if either $\satliteral_i^1$ or $\satliteral_i^2$ evaluates to $\valtrue$, and kept with 3 symbols otherwise; see \cref{tab:bottomup_dataset5_nphard_if} for a detailed simulation of why this is the case.
    We thus have:
    \begin{subequations} \label{eq:bottomup_nphard_if_proof_dataset_5}
    \begin{align}
        \sum_{\characters \in \dataset_3} |\bottomuptokenfull| 
        &= \sum_{\characters \in \dataset_3} \left(3 - 1\, \one\left\{
        \!\!\!
        \begin{array}{c}
            \satliteral_i^1 = \satvar_j \,\,\mathtt{and}\,\,
            \mergestring{\spacesymbol}{\satyesvarj\spacesymboltwo}, \mergestring{\spacesymbol\satyesvarj}{\spacesymbol} \in \merges  \\
            \mathrm{or} \\
            \satliteral_i^1 = \neg\satvar_j \,\,\mathtt{and}\,\,
            \mergestring{\spacesymboltwo\satnotvarj}{\spacesymbol}, \mergestring{\spacesymbol}{\satnotvarj\spacesymbol} \in \merges \\
            \mathrm{or} \\
            \satliteral_i^2 = \satvar_{j'} \,\,\mathtt{and}\,\,
            \mergestring{\spacesymbol}{\satyesvarjprime\spacesymboltwo}, \mergestring{\spacesymbol\satyesvarjprime}{\spacesymbol} \in \merges  \\
            \mathrm{or} \\
            \satliteral_i^2 = \neg\satvar_{j'} \,\,\mathtt{and}\,\,
            \mergestring{\spacesymboltwo\satnotvarjprime}{\spacesymbol}, \mergestring{\spacesymbol}{\satnotvarjprime\spacesymbol} \in \merges
        \end{array}
        \!\!
        \right\} \right) \\
        & = 3\satnclauses - \sum_{i=1}^{I} 
            \one\{\satliteral_i^1  \lor \satliteral_i^2 \} \\
        & \leq 3\satnclauses - \minclauses
    \end{align}
    \end{subequations}
    where, by construction, we have that a merge in our sequence (e.g., $\mergestring{\spacesymbol}{\satyesvarj\spacesymboltwo}$ or $\mergestring{\spacesymboltwo\satnotvarj}{\spacesymbol}$), if and only if its value is in a satisfying assignment (e.g., $\satvarsatisfiedj=\valtrue$ or $\satvarsatisfiedj=\valfalse$ respectively).
    Summing together the lengths in \cref{eq:bottomup_nphard_if_proof_dataset_1,eq:bottomup_nphard_if_proof_dataset_2,eq:bottomup_nphard_if_proof_dataset_34,eq:bottomup_nphard_if_proof_dataset_5}, we get that:
    \begin{align}
        \sum_{\characters \in \dataset} |\bottomuptokenfull| \leq \maxsymbols = (6\nrepeatvar + 6\nrepeatvar' + 4\nrepeatvar'' + 4\nrepeatvar''')\,\satnvariables + 3\,\satnclauses - \minclauses
    \end{align}
    which concludes the proof.

\newcommand{\samesubstring}{$\cdot$}

    \begin{table}[t]
        \centering
        \resizebox{\textwidth}{!}{%
        \begin{tabular}{cccccccccccc}
            \toprule
            $\dataset$ &
            $\characters$ 
            & $\bottomuptoken[\merges_1](\characters)$ 
            & \multicolumn{2}{c}{$\bottomuptoken[\merges_1 \circ \merges_2](\characters)$} 
            & \multicolumn{2}{c}{$\bottomuptoken[\merges_1 \circ \merges_2 \circ \merges_3](\characters)$} 
            & \multicolumn{2}{c}{$\bottomuptoken[\merges_1 \circ \merges_2 \circ \merges_3 \circ \merges_4](\characters)$}
            & \multicolumn{2}{c}{$\bottomuptoken[\merges_1 \circ \merges_2 \circ \merges_3 \circ \merges_5](\characters)$}
            & $|\bottomuptokenfull|$ \\
            \cmidrule(lr){4-5}
            \cmidrule(lr){6-7}
            \cmidrule(lr){8-9}
            \cmidrule(lr){10-11}
            &&& $\satvarsatisfiedj = \valtrue$ &$\satvarsatisfiedj = \valfalse$ 
            & $\satvarsatisfiedj = \valtrue$ &$\satvarsatisfiedj = \valfalse$ 
            & $\satvarsatisfiedj = \valtrue$ &$\satvarsatisfiedj = \valfalse$ 
            & $\satvarsatisfiedj = \valtrue$ &$\satvarsatisfiedj = \valfalse$ 
            \\
            \midrule
$\dataset_3$ &
$\subwordstring{\langle\spacesymbol,\satyesvarj,\spacesymbol,\satnotvarj,\spacesymbol\rangle}$ &
$\cdot$ & $\cdot$ & $\cdot$ & 
\multicolumn{2}{c}{$\subwordstring{\langle\spacesymbol\satyesvarj,\spacesymbol,\satnotvarj\spacesymbol\rangle}$} &
$\subwordstring{\langle\spacesymbol\satyesvarj\spacesymbol,\satnotvarj\spacesymbol\rangle}$ &
$\subwordstring{\langle\spacesymbol\satyesvarj,\spacesymbol\satnotvarj\spacesymbol\rangle}$ &
$\cdot$ & $\cdot$ & 2 \\
$\dataset_3$ &
$\subwordstring{\langle\spacesymboltwo,\satnotvarj,\spacesymbol,\satyesvarj,\spacesymboltwo\rangle}$ &
$\subwordstring{\langle\spacesymboltwo\satnotvarj,\spacesymbol,\satyesvarj\spacesymboltwo\rangle}$ &
$\subwordstring{\langle\spacesymboltwo\satnotvarj,\spacesymbol\satyesvarj\spacesymboltwo\rangle}$ &
$\subwordstring{\langle\spacesymboltwo\satnotvarj\spacesymbol,\satyesvarj\spacesymboltwo\rangle}$ &
$\cdot$ & $\cdot$ & $\cdot$ & $\cdot$ & $\cdot$ & $\cdot$ & 2 \\
$\dataset_4$ &
$\subwordstring{\langle\spacesymbol,\satnotvarj,\spacesymbol,\satyesvarj,\spacesymboltwo\rangle}$ &
$\subwordstring{\langle\spacesymbol,\satnotvarj,\spacesymbol,\satyesvarj\spacesymboltwo\rangle}$ &
$\subwordstring{\langle\spacesymbol,\satnotvarj,\spacesymbol\satyesvarj\spacesymboltwo\rangle}$ &
$\cdot$ &
$\cdot$ &
$\subwordstring{\langle\spacesymbol,\satnotvarj\spacesymbol,\satyesvarj\spacesymboltwo\rangle}$ &
$\cdot$ &
$\subwordstring{\langle\spacesymbol\satnotvarj\spacesymbol,\satyesvarj\spacesymboltwo\rangle}$ &
$\subwordstring{\langle\spacesymbol\satnotvarj,\spacesymbol\satyesvarj\spacesymboltwo\rangle}$ &
$\cdot$ & 2 \\
$\dataset_4$ &
$\subwordstring{\langle\spacesymboltwo,\satnotvarj,\spacesymbol,\satyesvarj,\spacesymbol\rangle}$ &
$\subwordstring{\langle\spacesymboltwo\satnotvarj,\spacesymbol,\satyesvarj,\spacesymbol\rangle}$ &
$\cdot$ &
$\subwordstring{\langle\spacesymboltwo\satnotvarj\spacesymbol,\satyesvarj,\spacesymbol\rangle}$ &
$\subwordstring{\langle\spacesymboltwo\satnotvarj,\spacesymbol\satyesvarj,\spacesymbol\rangle}$ &
$\cdot$ &
$\subwordstring{\langle\spacesymboltwo\satnotvarj,\spacesymbol\satyesvarj\spacesymbol\rangle}$ &
$\cdot$ & $\cdot$ &
$\subwordstring{\langle\spacesymboltwo\satnotvarj\spacesymbol,\satyesvarj\spacesymbol\rangle}$ &
2 \\
\bottomrule
        \end{tabular}
        }
        \caption{Example of applying $\merges$ in $\dataset_3$ and $\dataset_4$ of the bottom-up tokenisation problem obtained from \cref{reduction:max2set_to_bottomup_reduction}. The dot symbol $\cdot$ denotes the string not changing under the given merge.}
        \label{tab:bottomup_dataset3_nphard_if}
    \end{table}
    
    \begin{table}[t]
        \centering
        \resizebox{\textwidth}{!}{%
        \begin{tabular}{cccccccccccccccccccc}
            \toprule
            Assignment & Condition &
            $\characters$ 
            & $\bottomuptoken[\merges_1](\characters)$ 
            & $\bottomuptoken[\merges_1 \circ \merges_2](\characters)$
            & $\bottomuptoken[\merges_1 \circ \merges_2 \circ \merges_3](\characters)$
            & $\bottomuptoken[\merges_1 \circ \merges_2 \circ \merges_3 \circ \merges_4](\characters)$
            & $|\bottomuptokenfull|$ \\
            \midrule
\multirow{4}{*}{$\satliteral_i^1 = \satvar_j$ and $\satliteral_i^2 = \neg\satvar_{j'}$} & $\satvarsatisfiedj = \valtrue \land \satvarsatisfiedjprime = \valtrue$ &
\multirow{4}{*}{$\subwordstring{\langle\spacesymbol,\satyesvarj,\spacesymbol,\satnotvarjprime,\spacesymbol\rangle}$} &
\samesubstring & \samesubstring & 
\multirow{4}{*}{$\subwordstring{\langle\spacesymbol\satyesvarj,\spacesymbol,\satnotvarjprime\spacesymbol\rangle}$} &
$\subwordstring{\langle\spacesymbol\satyesvarj\spacesymbol,\satnotvarjprime\spacesymbol\rangle}$ &
2 \\
& $\satvarsatisfiedj = \valfalse \land \satvarsatisfiedjprime = \valtrue$ &&
\samesubstring & \samesubstring && 
$\subwordstring{\langle\spacesymbol\satyesvarj,\spacesymbol,\satnotvarjprime\spacesymbol\rangle}$ &
3 \\
& $\satvarsatisfiedj = \valtrue \land \satvarsatisfiedjprime = \valfalse$ &&
\samesubstring & \samesubstring && 
$\subwordstring{\langle\spacesymbol\satyesvarj\spacesymbol,\satnotvarjprime\spacesymbol\rangle}$ &
2 \\
& $\satvarsatisfiedj = \valfalse \land  \satvarsatisfiedjprime = \valfalse$ &&
\samesubstring & \samesubstring && 
$\subwordstring{\langle\spacesymbol\satyesvarj,\spacesymbol\satnotvarjprime\spacesymbol\rangle}$ & 
2 \\
\midrule
\multirow{4}{*}{$\satliteral_i^1 = \neg\satvar_j$ and $\satliteral_i^2 = \satvar_{j'}$}
& $\satvarsatisfiedj = \valtrue \land \satvarsatisfiedjprime = \valtrue$ &
\multirow{4}{*}{$\subwordstring{\langle\spacesymbol,\satyesvarjprime,\spacesymbol,\satnotvarj,\spacesymbol\rangle}$} &
\samesubstring & \samesubstring & 
\multirow{4}{*}{$\subwordstring{\langle\spacesymbol\satyesvarjprime,\spacesymbol,\satnotvarj\spacesymbol\rangle}$} &
$\subwordstring{\langle\spacesymbol\satyesvarjprime\spacesymbol,\satnotvarj\spacesymbol\rangle}$ &
2 \\
& $\satvarsatisfiedj = \valfalse \land \satvarsatisfiedjprime = \valtrue$ &&
\samesubstring & \samesubstring && 
$\subwordstring{\langle\spacesymbol\satyesvarjprime\spacesymbol,\satnotvarj\spacesymbol\rangle}$ & 
2 \\
& $\satvarsatisfiedj = \valtrue \land \satvarsatisfiedjprime = \valfalse$ &&
\samesubstring & \samesubstring && 
$\subwordstring{\langle\spacesymbol\satyesvarjprime,\spacesymbol,\satnotvarj\spacesymbol\rangle}$ & 
3 \\
& $\satvarsatisfiedj = \valfalse \land  \satvarsatisfiedjprime = \valfalse$ &&
\samesubstring & \samesubstring && 
$\subwordstring{\langle\spacesymbol\satyesvarjprime,\spacesymbol\satnotvarj\spacesymbol\rangle}$ & 
2 \\
\midrule
\multirow{4}{*}{$\satliteral_i^1 = \neg\satvar_j$ and $\satliteral_i^2 = \neg\satvar_{j'}$}
& $\satvarsatisfiedj = \valtrue \land \satvarsatisfiedjprime = \valtrue$ & 
\multirow{4}{*}{$\subwordstring{\langle\spacesymboltwo,\satnotvarj,\spacesymbol,\satnotvarjprime,\spacesymbol\rangle}$} &
\multirow{4}{*}{$\subwordstring{\langle\spacesymboltwo\satnotvarj,\spacesymbol,\satnotvarjprime,\spacesymbol\rangle}$} & 
\samesubstring & 
$\subwordstring{\langle\spacesymboltwo\satnotvarj,\spacesymbol,\satnotvarjprime\spacesymbol\rangle}$ &
\samesubstring & 3
\\
& $\satvarsatisfiedj = \valfalse \land \satvarsatisfiedjprime = \valtrue$ &&& 
$\subwordstring{\langle\spacesymboltwo\satnotvarj\spacesymbol,\satnotvarjprime,\spacesymbol\rangle}$ &
$\subwordstring{\langle\spacesymboltwo\satnotvarj\spacesymbol,\satnotvarjprime\spacesymbol\rangle}$ &
\samesubstring & 2
\\
& $\satvarsatisfiedj = \valtrue \land \satvarsatisfiedjprime = \valfalse$ &&& 
\samesubstring & 
$\subwordstring{\langle\spacesymboltwo\satnotvarj,\spacesymbol,\satnotvarjprime\spacesymbol\rangle}$ &
$\subwordstring{\langle\spacesymboltwo\satnotvarj,\spacesymbol\satnotvarjprime\spacesymbol\rangle}$ & 2
\\
& $\satvarsatisfiedj = \valfalse \land  \satvarsatisfiedjprime = \valfalse$ &&& 
$\subwordstring{\langle\spacesymboltwo\satnotvarj\spacesymbol,\satnotvarjprime,\spacesymbol\rangle}$ &
$\subwordstring{\langle\spacesymboltwo\satnotvarj\spacesymbol,\satnotvarjprime\spacesymbol\rangle}$ &
\samesubstring & 2 \\
\midrule
\multirow{4}{*}{$\satliteral_i^1 = \satvar_j$ and $\satliteral_i^2 = \satvar_{j'}$}
& $\satvarsatisfiedj = \valtrue \land \satvarsatisfiedjprime = \valtrue$ & 
\multirow{4}{*}{$\subwordstring{\langle\spacesymbol,\satyesvarj,\spacesymbol,\satyesvarjprime,\spacesymboltwo\rangle}$} & 
\multirow{4}{*}{$\subwordstring{\langle\spacesymbol,\satyesvarj,\spacesymbol,\satyesvarjprime\spacesymboltwo\rangle}$} & 
\multirow{2}{*}{$\subwordstring{\langle\spacesymbol,\satyesvarj,\spacesymbol\satyesvarjprime\spacesymboltwo\rangle}$} &
\multirow{2}{*}{$\subwordstring{\langle\spacesymbol\satyesvarj,\spacesymbol\satyesvarjprime\spacesymboltwo\rangle}$} &
\multirow{2}{*}{$\subwordstring{\langle\spacesymbol\satyesvarj,\spacesymbol\satyesvarjprime\spacesymboltwo\rangle}$} & 2
\\
& $\satvarsatisfiedj = \valfalse \land \satvarsatisfiedjprime = \valtrue$ &&&&&& 2 \\
& $\satvarsatisfiedj = \valtrue \land \satvarsatisfiedjprime = \valfalse$ &&& 
\multirow{2}{*}{\samesubstring} & 
\multirow{2}{*}{$\subwordstring{\langle\spacesymbol\satyesvarj,\spacesymbol,\satyesvarjprime\spacesymboltwo\rangle}$} & 
$\subwordstring{\langle\spacesymbol\satyesvarj\spacesymbol,\satyesvarjprime\spacesymboltwo\rangle}$ & 2
\\
& $\satvarsatisfiedj = \valfalse \land  \satvarsatisfiedjprime = \valfalse$ &&&&& 
\samesubstring & 3 \\
\bottomrule
        \end{tabular}
        }
        \caption{Example of applying $\merges$ in $\dataset_5$ of the bottom-up tokenisation problem obtained from \cref{reduction:max2set_to_bottomup_reduction}. The dot symbol $\cdot$ denotes the string not changing under the given merge.}
        \label{tab:bottomup_dataset5_nphard_if}
    \end{table}
\end{proof}

\section{Proof of \texorpdfstring{\Cref{lemma:bottomup_nphard_onlyif_lemma}}{Lemma}}
\label{appendix:proof_bottomup_nphard_onlyif_lemma}

\bottomupnphardonlyiflemma*
\begin{proof}
    First, note that:
    \begin{align}
        \sum_{\characters \in \dataset} |\characters| = (12\nrepeatvar + 12\nrepeatvar' + 10\nrepeatvar'' + 10\nrepeatvar''')\,\satnvariables + 5\,\satnclauses
    \end{align}
    Further, let:
    \begin{gather}
    	\tokeniselengthmerge \defeq \sum_{\characters \in \dataset} |\bottomuptokenfull|,
    	\qquad\qquad
    	\mergebase = \merges_1 \circ \merges_2 \circ \merges_3 \circ \merges_4 \circ \merges_5 \\
        \merges_1 = \bigcirc_{j=1}^{\satnvariables} [\mergestring{\spacesymboltwo}{\satnotvarj}, \mergestring{\satyesvarj}{\spacesymboltwo}], \qquad
        \merges_3 = \bigcirc_{j=1}^{\satnvariables} [\mergestring{\satnotvarj}{\spacesymbol}, \mergestring{\spacesymbol}{\satyesvarj}], \qquad
        \merges_5 = \bigcirc_{j=1}^{\satnvariables} [\mergestring{\spacesymbol}{\satnotvarj}, \mergestring{\satyesvarj}{\spacesymbol}] \nonumber \\
        \merges_2 = \bigcirc_{j=1}^{\satnvariables} \left[
        \begin{array}{lr}
            \mergestring{\spacesymbol}{\satyesvarj\spacesymboltwo} & \mathtt{if}\,\,\satvarsatisfiedj=\valtrue \\
            \mergestring{\spacesymboltwo\satnotvarj}{\spacesymbol} & \mathtt{else}
        \end{array}
        \right], \qquad
        \merges_4 = \bigcirc_{j=1}^{\satnvariables} \left[
        \begin{array}{lr}
            \mergestring{\spacesymbol\satyesvarj}{\spacesymbol} & \mathtt{if}\,\,\satvarsatisfiedj=\valtrue \\
            \mergestring{\spacesymbol}{\satnotvarj\spacesymbol} & \mathtt{else}
        \end{array}
        \right] \nonumber \\
        \mergesthreecharsone = \bigg\{
\!\!\begin{array}{c}
\mergestring{\spacesymbol\satyesvarj}{\spacesymbol}, \mergestring{\spacesymbol}{\satnotvarj\spacesymbol} \\
\mergestring{\spacesymbol}{\satyesvarj\spacesymbol}, \mergestring{\spacesymbol\satnotvarj}{\spacesymbol}
\end{array}\!\!
        \bigg\},\qquad
        \mergesthreecharstwo = \bigg\{\!\!
\begin{array}{c}
\mergestring{\spacesymbol}{\satyesvarj\spacesymboltwo}, \mergestring{\spacesymboltwo\satnotvarj}{\spacesymbol}\\
\mergestring{\spacesymbol\satyesvarj}{\spacesymboltwo}, \mergestring{\spacesymboltwo}{\satnotvarj\spacesymbol} 
\end{array}\!\!\bigg\} \nonumber \\
        \mergesthreecharstrue = \bigg\{
\!\!\begin{array}{c}
\mergestring{\spacesymbol\satyesvarj}{\spacesymbol}, 
\mergestring{\spacesymbol}{\satyesvarj\spacesymboltwo}, \\
\mergestring{\spacesymbol}{\satyesvarj\spacesymbol}, 
\mergestring{\spacesymbol\satyesvarj}{\spacesymboltwo}, 
\end{array}\!\!
        \bigg\},\qquad
        \mergesthreecharsfalse = \bigg\{\!\!
\begin{array}{c}
\mergestring{\spacesymbol}{\satnotvarj\spacesymbol},
\mergestring{\spacesymboltwo\satnotvarj}{\spacesymbol}, \\
\mergestring{\spacesymbol\satnotvarj}{\spacesymbol},
\mergestring{\spacesymboltwo}{\satnotvarj\spacesymbol},
\end{array}\!\!\bigg\} \nonumber
    \end{gather}
    The maximum number of symbols in this dataset after compression is set to $\maxsymbols = (6\nrepeatvar \!+\! 6\nrepeatvar' \!+\! 4\nrepeatvar'' \!+\! 4\nrepeatvar''')\,\satnvariables \!+\! 3\,\satnclauses \!-\! \minclauses$.
    This means that to satisfy this objective, there must exist a vocabulary whose tokeniser compresses the text by at least $(6\nrepeatvar + 6\nrepeatvar' + 6\nrepeatvar'' + 6\nrepeatvar''')\,\satnvariables + 2\satnclauses + \minclauses$ symbols.
    We now prove this lemma in five steps:
    \circled{1} we show that any solution which compresses the text by at least $6\nrepeatvar\satnvariables$
    symbols must include all merges in $\merges_1$, $\merges_3$, and $\merges_5$;
    \circled{2} we show that any solution which compresses the text by at least $(6\nrepeatvar  + 6\nrepeatvar')\satnvariables$ symbols must only include either merges in $\merges_1$, $\merges_3$, $\merges_5$, or in either $\mergesthreecharsone$ or $\mergesthreecharstwo$;
    \circled{3} we show that any solution which compresses the text by at least $(6\nrepeatvar + 6\nrepeatvar' + 6\nrepeatvar'')\satnvariables$ symbols must include, for each $j \in \{1, \satnvariables\}$, exactly one merge in set $\mergesthreecharsone$ and one in set $\mergesthreecharstwo$;
    \circled{4} we show that any solution which compresses the text by at least $(6\nrepeatvar + 6\nrepeatvar' + 6\nrepeatvar'' + 6\nrepeatvar''')\satnvariables$ symbols must include, for each $j \in \{1, \satnvariables\}$, exactly two merge in either set $\mergesthreecharstrue$ or in set $\mergesthreecharsfalse$;
    \circled{5} we show that any solution which compresses the text by at least $(6\nrepeatvar  + 6\nrepeatvar' + 4\nrepeatvar'' + 4\nrepeatvar''')\satnvariables + 6\nrepeatvar'' + 6\nrepeatvar'')\satnvariables + 2\satnclauses + \minclauses$
    symbols must be produced by a \maxsatacron problem with at least $\minclauses$ satisfied clauses.
\end{proof}

\newcommand{\mathcomment}[1]{\textcolor{gray}{\text{#1}}}

    \begin{mylemmastep} \textnormal{(Step \circled{1}).}
    Any solution which compresses the text by at least $6\nrepeatvar\satnvariables$ symbols must include all merges in $\merges_1$, $\merges_3$, and $\merges_5$, i.e.,:
    \begin{align}
        &\bigg(\tokeniselengthmerge \leq \underbrace{6\nrepeatvar\satnvariables + (12\nrepeatvar' + 10\nrepeatvar'' + 10\nrepeatvar''')\,\satnvariables + 5\,\satnclauses}_{\sum_{\characters \in \dataset} |\characters| - 6\nrepeatvar\satnvariables}\bigg) 
        \\
        &\quad \implies
        \underbrace{\bigcirc_{j=1}^{\satnvariables} [\mergestring{\spacesymboltwo}{\satnotvarj}, \mergestring{\satyesvarj}{\spacesymboltwo}]}_{\merges_1} \subset \merges, 
        \quad
        \underbrace{\bigcirc_{j=1}^{\satnvariables} [\mergestring{\satnotvarj}{\spacesymbol}, \mergestring{\spacesymbol}{\satyesvarj}]}_{\merges_3} \subset \merges, 
        \quad
        \underbrace{\bigcirc_{j=1}^{\satnvariables} [\mergestring{\spacesymbol}{\satnotvarj}, \mergestring{\satyesvarj}{\spacesymbol}]}_{\merges_5} \subset \merges \nonumber
    \end{align}
    \end{mylemmastep}
    \begin{proof}
    We prove this statement by contradiction.
    Assume that one of the subwords above is not present in the tokenisers' merge sequence $\merges$.
    In that case, the strings in $\dataset_1$ which contain this character string will not be compressed, and will thus still be represented with 2 symbols.
    There will thus be at most $6\satnvariables-1$ strings in $\dataset_1$ represented with a single symbol, and at least one represented with two symbols.
    The minimum length achievable would thus be:
    \begin{subequations}
    \begin{align}
        \tokeniselengthmerge 
        &= \underbrace{\sum_{\characters \in \bigcup_{\_ = 1}^{\nrepeatvar} \dataset_0} |\bottomuptokenfull|}_{\geq (6\satnvariables-1)\nrepeatvar + 2\nrepeatvar} + \underbrace{\sum_{\characters \in \dataset\setminus(\bigcup_{\_ = 1}^{\nrepeatvar} \dataset_0)} |\bottomuptokenfull|}_{>0} \\
		&> (6 \satnvariables + 1)\nrepeatvar 
        &\!\!\!\!\!\!\!\!\!\!\!\!\!\!\!\!\!\!\!\!\!\!\!\!\!\!\!\!\!\!\!\!\!\!\!\!\!\!\!\!\!\!\!\!\!\!\!\!\!\!\!\!\!\!\!\!\!\!\!\!\!\!\!\!\!\!\!\!\!\!\!\!\!
        \mathcomment{By construction $\nrepeatvar = (12\nrepeatvar' + 10\nrepeatvar'' + 10\nrepeatvar''')\,\satnvariables + 5\satnclauses$} \\
		&= (6\nrepeatvar + 12\nrepeatvar' + 10\nrepeatvar'' + 10\nrepeatvar''')\,\satnvariables + 5\satnclauses 
    \end{align}
    \end{subequations}
    which contradicts the proofs statement.
    \end{proof}
    
    \begin{mylemmastep} \textnormal{(Step \circled{2}).}
    Any solution which compresses the text by at least $(6\nrepeatvar  + 6\nrepeatvar')\satnvariables$ symbols must only include either merges in $\merges_1$, $\merges_3$, $\merges_5$, or in either $\mergesthreecharsone$ or $\mergesthreecharstwo$, i.e.,:
    \begin{align}
        &\bigg(\tokeniselength \leq \underbrace{(6\nrepeatvar + 6\nrepeatvar')\satnvariables + (10\nrepeatvar'' + 10\nrepeatvar''')\,\satnvariables + 5\,\satnclauses
        }_{\sum_{\characters \in \dataset} |\characters| - (6\nrepeatvar + 6\nrepeatvar')\satnvariables}\bigg)  
       \\ &\qquad\qquad\qquad\qquad
        \implies 
        \merges \setminus(\merges_1\circ\merges_3\circ\merges_5) \subseteq 
\underbrace{\bigg\{\!\!
\begin{array}{c}
\mergestring{\spacesymbol}{\satyesvarj\spacesymboltwo}, \mergestring{\spacesymboltwo\satnotvarj}{\spacesymbol}, \mergestring{\spacesymbol\satyesvarj}{\spacesymbol}, \mergestring{\spacesymbol}{\satnotvarj\spacesymbol} \\
\mergestring{\spacesymbol\satyesvarj}{\spacesymboltwo}, \mergestring{\spacesymboltwo}{\satnotvarj\spacesymbol}, \mergestring{\spacesymbol}{\satyesvarj\spacesymbol}, \mergestring{\spacesymbol\satnotvarj}{\spacesymbol}
\end{array}
\!\!\bigg\}_{j=1}^{\satnvariables}
}_{\bigcup_{j=1}^{\satnvariables} (\mergesthreecharsone \cup \mergesthreecharstwo)}
\nonumber
    \end{align}
    \end{mylemmastep}
    \begin{proof}
    We again prove this statement by contradiction.
    Assume that $\merges$ has all merges $\merges_1, \merges_3, \merges_5$, but one of its other merges is in neither of the sets $\mergesthreecharsone$ and $\mergesthreecharstwo$.
    This means that at least one of the sets $\mergesthreecharsone$ and $\mergesthreecharstwo$ will have no merge in the solution; this is because there are $2\satnvariables$ such sets, which---coupled together with the  $6\satnvariables$ already selected merges in $\merges_1, \merges_3, \merges_5$---would amount to the maximum of $8\satnvariables$ merges.
    In that case, the strings (e.g.,
$\charstring{\spacesymbol\satyesvarj\spacesymbol}$,
$\charstring{\spacesymbol\satnotvarj\spacesymbol}$,
$\charstring{\spacesymbol\satyesvarj\spacesymboltwo}$ and
$\charstring{\spacesymboltwo\satnotvarj\spacesymbol}$) in $\dataset_2$ containing the characters this absent merge represents will not be fully compressed to a single symbol, being represented with 2 symbols instead.
    There will thus be at most $6\satnvariables-1$ strings in $\dataset_2$ represented with a single symbol, and at least one represented with two symbols.
    The minimum length achievable would thus be:
    \begin{subequations}
    \begin{align}
        \tokeniselengthmerge 
        &= \underbrace{\sum_{\characters \in \bigcup_{\_ = 1}^{\nrepeatvar} \dataset_1} |\bottomuptokenfull|}_{= 6\nrepeatvar\satnvariables} + 
        \underbrace{\sum_{\characters \in \bigcup_{\_ = 1}^{\nrepeatvar'} \dataset_2} |\bottomuptokenfull|}_{\geq (6\satnvariables-1)\nrepeatvar' + 2\nrepeatvar'} + 
        \underbrace{\sum_{\characters \in \dataset\setminus( \dataset_1\cup\dataset_2)} |\bottomuptokenfull|}_{>0} 
        \!\!\!\!\!\!\!\!\!\!\!\!\!\!\!\!\!\!\!\!\!\!\!\!\!\!\!\!\!\!\!\!\!\!\!\!\!\!\!\!\!\!\!\!\!\!\!\!\!\!\!\!\!\!\!\!\!\!\!\!\!\!\!\!\!\!\!\!\!\!\!\!\!\!\!\!\!\!\!\!\!\!\!\!\!\!\!\!\!\!\!\!\! \\
		&> 6\nrepeatvar\satnvariables + (6 \satnvariables + 1)\nrepeatvar' 
        &\!\!\!\!\!\!\!\!\!\!\!\!\!\!\!\!\!\!\!\!\!\!\!\!\!\!\!\!\!\!\!\!\!\!\!
        \mathcomment{By construction $\nrepeatvar' = (10\nrepeatvar'' + 10\nrepeatvar''')\,\satnvariables + 5\satnclauses$} \\
		&= (6\nrepeatvar + 6\nrepeatvar' + 10\nrepeatvar'' + 10\nrepeatvar''')\,\satnvariables + 5\,\satnclauses 
    \end{align}
    \end{subequations}
    which contradicts the proofs statement.
    \end{proof}

    \begin{mylemmastep} \textnormal{(Step \circled{3}).}
    Any solution which compresses the text by at least $(6\nrepeatvar + 6\nrepeatvar' + 6\nrepeatvar'')\satnvariables$ symbols must include all merges in $\merges_1$, $\merges_3$, $\merges_5$, and, for each $j \in \{1, \satnvariables\}$, exactly one merge in set $\mergesthreecharsone$ and one in set $\mergesthreecharstwo$, i.e.,:
    \begin{align}
        &\bigg(\tokeniselengthmerge \leq \underbrace{(6\nrepeatvar + 6\nrepeatvar' + 4\nrepeatvar'')\satnvariables + 10\nrepeatvar''' \,\satnvariables + 5\,\satnclauses 
        }_{\sum_{\characters \in \dataset} |\characters| - (6\nrepeatvar + 6\nrepeatvar' + 6\nrepeatvar'') \satnvariables}\bigg)  
        \\ &\qquad\quad
        \implies 
        \forall_{j \in \{1, ..., \satnvariables\}}\,\,
        \bigg|\merges \cap \underbrace{\bigg\{
\!\!\begin{array}{c}
\mergestring{\spacesymbol\satyesvarj}{\spacesymbol}, \mergestring{\spacesymbol}{\satnotvarj\spacesymbol} \\
\mergestring{\spacesymbol}{\satyesvarj\spacesymbol}, \mergestring{\spacesymbol\satnotvarj}{\spacesymbol}
\end{array}\!\!
        \bigg\}}_{\mergesthreecharsone}\bigg| = 1 
        \,\,\mathtt{and}\,\,
        \bigg|\merges \cap \underbrace{\bigg\{\!\!
\begin{array}{c}
\mergestring{\spacesymbol}{\satyesvarj\spacesymboltwo}, \mergestring{\spacesymboltwo\satnotvarj}{\spacesymbol}\\
\mergestring{\spacesymbol\satyesvarj}{\spacesymboltwo}, \mergestring{\spacesymboltwo}{\satnotvarj\spacesymbol} 
\end{array}\!\!\bigg\}}_{\mergesthreecharstwo}\bigg| = 1
        \nonumber
    \end{align}
    \end{mylemmastep}
    \begin{proof}
    We again prove this statement by contradiction.
    First, assume that $\merges$ contains all the merges in $\merges_1, \merges_3, \merges_5$; further, assume all its other merges are contained in sets 
    $\mergesthreecharsone$ and $\mergesthreecharstwo$.
    Note now that, if any merge in $\mergesthreecharstwo$ is in the selected merges $\merges$, the string $\charstring{\spacesymboltwo\satnotvarj\spacesymbol\satyesvarj\spacesymboltwo}$ in $\dataset_3$ will be compressed to 2 symbols (e.g., $\subwordstring{\langle\spacesymboltwo\satnotvarj,\spacesymbol\satyesvarj\spacesymboltwo\rangle}$);
    if none of these merges is present, however, this string will only be compressed to 3 symbols  (e.g., $\subwordstring{\langle\spacesymboltwo\satnotvarj,\spacesymbol,\satyesvarj\spacesymboltwo\rangle}$).
    The same is true for strings $\charstring{\spacesymbol\satyesvarj\spacesymbol\satnotvarj\spacesymbol}$ and merges in $\mergesthreecharsone$.
    Now, assume the contradictory case: for a value of $j \in \{1, \satnvariables\}$, $\merges$ does not satisfy the condition above.
    As, by construction, our solution has $\vocabsize = 8\satnvariables$ merges, and because $|\merges_1 \circ \merges_3 \circ \merges_5| = 6\satnvariables$, we know that we have $2\satnvariables$ merges in sets $\mergesthreecharsone$ and $\mergesthreecharstwo$.
    As there are exactly $2\satnvariables$ such sets, if the condition above does not hold, at least one of these sets must have no merge present in $\merges$.
    In that case, the strings in $\dataset_3$ which contain the character string represented by these absent merges will be compressed to three symbols, while others will be compressed to two symbols.
    There will thus be at most $2\satnvariables-1$ strings in $\dataset_3$ represented with two symbols, and at least one represented with three symbols.
    The minimum length achievable would thus be:
    \begin{subequations}
    \begin{align}
        \tokeniselengthmerge 
        &= \underbrace{\sum_{\characters \in \bigcup\limits_{\_ = 1}^{\nrepeatvar} \dataset_1 \cup \bigcup\limits_{\_ = 1}^{\nrepeatvar'} \dataset_2} \!\!\!\!\!\!\!\!\!\!\!\! 
        |\bottomuptokenfull|}_{= (6\nrepeatvar+6\nrepeatvar')\satnvariables} + 
        \underbrace{\sum_{\characters \in \bigcup\limits_{\_ = 1}^{\nrepeatvar''} \dataset_3} |\bottomuptokenfull|}_{\geq (2\satnvariables-1)2\nrepeatvar'' + 3\nrepeatvar''} + 
        \underbrace{\sum_{\characters \in 
        \bigcup\limits_{\_ = 1}^{\nrepeatvar'''} \dataset_4\cup
        \dataset_5}\!\!\!\!\!\!\!
        |\bottomuptokenfull|}_{>0} 
        \\
		&> (6\nrepeatvar+6\nrepeatvar')\satnvariables + (4 \satnvariables + 1)\nrepeatvar'' 
        &\!\!\!\!\!\!\!\!\!\!\!\!\!\!\!\!\!\!\!\!\!\!\!\!\!\!\!\!\!\!\!\!\!\!\!\!\!\!\!\!\!\!\!\!\!\!\!\!\!\!\!\!\!\!\!\!\!\!\!\!\!\!\!\!\!\!\!\!\!\!\!\!\!\! 
        \mathcomment{By construction $\nrepeatvar'' = 10\nrepeatvar'''\,\satnvariables + 5\satnclauses$} \\
		&= (6\nrepeatvar + 6\nrepeatvar' + 4\nrepeatvar'' + 10\nrepeatvar''')\,\satnvariables + 5\,\satnclauses 
    \end{align}
    \end{subequations}
    which contradicts the proof's statement.
    \end{proof}

    \begin{mylemmastep} \textnormal{(Step \circled{4}).}
    Any solution which compresses the text by at least $(6\nrepeatvar + 6\nrepeatvar' + 6\nrepeatvar'' + 6\nrepeatvar''')\satnvariables$ symbols must include all merges in $\merges_1$, $\merges_3$, $\merges_5$, and,
    for each $j \in \{1, \satnvariables\}$, exactly one merge in set $\mergesthreecharsone$ and one in set $\mergesthreecharstwo$, such that either both these merges are in $\mergesthreecharstrue$ or both are in $\mergesthreecharsfalse$, i.e.,:
    \begin{align}
        &\bigg(\tokeniselengthmerge \leq \underbrace{(6\nrepeatvar + 6\nrepeatvar' + 4\nrepeatvar'' + 4\nrepeatvar''')\satnvariables + 5\,\satnclauses 
        }_{\sum_{\characters \in \dataset} |\characters| - (6\nrepeatvar + 6\nrepeatvar' + 6\nrepeatvar'' + 6\nrepeatvar''') \satnvariables}\bigg)  
        \\ &\qquad\quad
        \implies 
        \forall_{j \in \{1, ..., \satnvariables\}}\,\,
        |\merges \cap \underbrace{\bigg\{
\!\!\begin{array}{c}
\mergestring{\spacesymbol\satyesvarj}{\spacesymbol}, 
\mergestring{\spacesymbol}{\satyesvarj\spacesymboltwo}, \\
\mergestring{\spacesymbol}{\satyesvarj\spacesymbol}, 
\mergestring{\spacesymbol\satyesvarj}{\spacesymboltwo}, 
\end{array}\!\!
        \bigg\}}_{\mergesthreecharstrue}| = 2 
        \,\,\mathtt{or}\,\,
        |\merges \cap \underbrace{\bigg\{\!\!
\begin{array}{c}
\mergestring{\spacesymbol}{\satnotvarj\spacesymbol},
\mergestring{\spacesymboltwo\satnotvarj}{\spacesymbol}, \\
\mergestring{\spacesymbol\satnotvarj}{\spacesymbol},
\mergestring{\spacesymboltwo}{\satnotvarj\spacesymbol},
\end{array}\!\!\bigg\}}_{\mergesthreecharsfalse}| = 2 \nonumber \\
    \end{align}
    \end{mylemmastep}
    \begin{proof}
    First, note that the conditions of the step of our proof are stricter than previous ones, so we assume the conditions of steps \circled{1} to \circled{3} hold---i.e., $\merges$ contains all merges in $\merges_1, \merges_3, \merges_5$; further, it has one and only one merge from each set $\mergesthreecharsone$ and $\mergesthreecharstwo$.
    (Note that $\mergesthreecharsone \cup \mergesthreecharstwo = \mergesthreecharstrue \cup \mergesthreecharsfalse$, and that the just-mentioned condition implies $|\merges \cap (\mergesthreecharstrue \cup \mergesthreecharsfalse)| = 2$.)
    We now again prove this statement by contradiction.
    Consider now the case:
    \begin{align}
        \bigg|\merges \cap \underbrace{\bigg\{
\!\!\begin{array}{c}
\mergestring{\spacesymbol\satyesvarj}{\spacesymbol}, 
\mergestring{\spacesymbol}{\satyesvarj\spacesymboltwo}, \\
\mergestring{\spacesymbol}{\satyesvarj\spacesymbol}, 
\mergestring{\spacesymbol\satyesvarj}{\spacesymboltwo}, 
\end{array}\!\!
        \bigg\}}_{\mergesthreecharstrue}\bigg| = 2 
        \,\,\mathtt{or}\,\,
        \bigg|\merges \cap \underbrace{\bigg\{\!\!
\begin{array}{c}
\mergestring{\spacesymbol}{\satnotvarj\spacesymbol},
\mergestring{\spacesymboltwo\satnotvarj}{\spacesymbol}, \\
\mergestring{\spacesymbol\satnotvarj}{\spacesymbol},
\mergestring{\spacesymboltwo}{\satnotvarj\spacesymbol},
\end{array}\!\!\bigg\}}_{\mergesthreecharsfalse}\bigg| = 2 
    \end{align}
    If this is true, then strings $\charstring{\spacesymbol\satnotvarj\spacesymbol\satyesvarj\spacesymboltwo}$ and $\charstring{\spacesymboltwo\satnotvarj\spacesymbol\satyesvarj\spacesymbol}$ in $\dataset_4$ will be compressed to 2 symbols each (e.g., to $\subwordstring{\langle\spacesymbol\satnotvarj,\spacesymbol\satyesvarj\spacesymboltwo\rangle}$ and $\subwordstring{\langle\spacesymboltwo\satnotvarj,\spacesymbol\satyesvarj\spacesymbol\rangle}$ 
    or $\subwordstring{\langle\spacesymbol\satnotvarj\spacesymbol,\satyesvarj\spacesymboltwo\rangle}$ and $\subwordstring{\langle\spacesymboltwo\satnotvarj\spacesymbol,\satyesvarj\spacesymbol\rangle}$ );
    if this condition is false, however, one of these strings will only be compressed to 3 symbols (e.g., to $\subwordstring{\langle\spacesymbol\satnotvarj,\spacesymbol\satyesvarj\spacesymboltwo\rangle}$ and $\subwordstring{\langle\spacesymboltwo\satnotvarj,\spacesymbol,\satyesvarj\spacesymbol\rangle}$).
    Now, assume the contradictory case: for a value of $j \in \{1, \satnvariables\}$, $\merges$ does not satisfy the condition above.
    In that case, the strings in $\dataset_4$ for which the condition does not hold will be compressed to $3+2$  symbols, while others will be compressed to $2+2$ symbols.
    There will thus be at most $2\satnvariables-1$ strings in $\dataset_4$ represented with two symbols, and at least one represented with three symbols.
    The minimum length achievable would thus be:
    \begin{subequations}
    \begin{align}
        \tokeniselengthmerge 
        &= \underbrace{\sum_{\characters \in \bigcup\limits_{\_ = 1}^{\nrepeatvar} \dataset_1 \cup \bigcup\limits_{\_ = 1}^{\nrepeatvar'} \dataset_2\cup \bigcup\limits_{\_ = 1}^{\nrepeatvar''} \dataset_3} \!\!\!\!\!\!\!\!\!\!\!\! 
        |\bottomuptokenfull|}_{= (6\nrepeatvar+6\nrepeatvar'+4\nrepeatvar'')\satnvariables} + 
        \underbrace{\sum_{\characters \in \bigcup\limits_{\_ = 1}^{\nrepeatvar'''} \dataset_4} |\bottomuptokenfull|}_{\geq (2\satnvariables-1)2\nrepeatvar''' + 3\nrepeatvar'''} + 
        \underbrace{\sum_{\characters \in 
        \dataset_5}
        |\bottomuptokenfull|}_{>0} 
        \\
		&> (6\nrepeatvar+6\nrepeatvar'+4\nrepeatvar'')\satnvariables + (4 \satnvariables + 1)\nrepeatvar''' 
        &\!\!\!\!\!\!\!\!\!\!\!\!\!\!\!\!\!\!\!\!\!\!\!\!\!\!\!\!\!\!\!\!\!\!\!\!\!\!\!\!\!\!\!\!\!\!\!\!\!\!\!\!\!\!\!\!\!\!\!\!\!\!\!\!\!\!\!\!\!\!\!\!\!\! 
        \mathcomment{By construction $\nrepeatvar''' = 5\satnclauses$\,\,} \\
		&= (6\nrepeatvar + 6\nrepeatvar' + 4\nrepeatvar'' + 4\nrepeatvar''')\,\satnvariables + 5\,\satnclauses 
    \end{align}
    \end{subequations}
    which contradicts the proof's statement.
    \end{proof}

    \begin{mylemmastep} \textnormal{(Step \circled{5}).}
    Any tokenisation problem with a solution which compresses the text by at least $(6\nrepeatvar  + 6\nrepeatvar' + 6\nrepeatvar'' + 6\nrepeatvar'')\satnvariables + 2\satnclauses + \minclauses$ symbols must be produced by a \maxsatacron problem with at least $\minclauses$ satisfied clauses, i.e.,:
    \begin{align}
        &\bigg(\tokeniselength \leq \underbrace{(6\nrepeatvar + 6\nrepeatvar' + 4\nrepeatvar'' + 4\nrepeatvar''')\satnvariables + 3\satnclauses - \minclauses 
        }_{\sum_{\characters \in \dataset} |\characters| - (6\nrepeatvar + 6\nrepeatvar' + 6\nrepeatvar'' + 6\nrepeatvar''') \satnvariables + 2\satnclauses + \minclauses}\bigg)  
        \implies 
        \maxsat(\satvars, \satclauses, \minclauses) \nonumber
    \end{align} 
    \end{mylemmastep}
    \begin{proof}
    Finally, we now know any solution with this compression must have---for any variable $\satvar_j$---either two merges in $\mergesthreecharstrue$ or in $\mergesthreecharsfalse$ (and never both).
    We can thus create a bijection $\bijectionmergesat$ between the set of possible merge sequences respecting this condition, and the set of $\valtrue$/$\valfalse$ assignments to SAT variables $\satvals$:
    \begin{align}
        \bijectionmergesat(\merges) = \bigg\{
        \begin{array}{lr}
            \valtrue & \mathtt{if}\,|\merges \cap \mergesthreecharstrue|=2 \\
            \valfalse & \mathtt{if}\,|\merges \cap \mergesthreecharsfalse|=2
        \end{array}
        \bigg\}_{j=1}^{\satnvariables}
    \end{align}
    Further, note that merge sequences of this form (as shown in \cref{eq:direct_compression_datasets12}) lead to exactly $(6\nrepeatvar + 6\nrepeatvar' + 6\nrepeatvar'' + 6\nrepeatvar''')\satnvariables$ symbols being compressed in datasets $\dataset_1$ to $\dataset_4$.
    To achieve the target compression, a solution must thus compress $\dataset_5$ by at least $2\satnclauses + \minclauses$ symbols.
    Now note that for any string, e.g., $\charstring{\spacesymbol\satliteral_i^1\spacesymbol\satliteral_i^2\spacesymbol}$, 
    in $\dataset_5$ we have three compression options: $\charstring{\spacesymbol\satliteral_i^1\spacesymbol}$ and $\charstring{\satliteral_i^2\spacesymbol}$ will be compressed, saving 3 symbols; 
    $\charstring{\spacesymbol\satliteral_i^1}$ and $\charstring{\spacesymbol\satliteral_i^2\spacesymbol}$  will be compressed, also saving 3 symbols; or only 
    $\charstring{\spacesymbol\satliteral_i^1}$ and $\charstring{\satliteral_i^2\spacesymbol}$ will be compressed saving only 2 symbols.
    More specifically,
    $\charstring{\spacesymbol\satliteral_i^1\spacesymbol}$ will be compressed
    if $\satliteral_i^1$ represents $\satvar_j$ and merge $\mergestring{\spacesymbol}{\satyesvarj\spacesymbol}$ exists,
    or if $\satliteral_i^1$ represents $\neg\satvar_j$ and subword $\mergestring{\spacesymbol\satnotvarj}{\spacesymbol}$ exists; the same is true for $\charstring{\spacesymbol\satliteral_i^2\spacesymbol}$.
    They cannot both be compressed, however, as there is only one symbol $\charstring{\spacesymbol}$ between the literals.
    We thus get a compression of $3$ symbols for each of these strings if at least one of its literals has an associated merge in $\merges$.
    Note thus that whenever a string $\charstring{\spacesymbol\satliteral_i^1\spacesymbol\satliteral_i^2\spacesymbol}$ is compressed by 3 symbols using vocabulary $\vocab$, the \maxsatacron disjunction $\satliteral_i^1 \lor \satliteral_i^2$ will also be satisfied by assignment $\satvals = \bijectionmergesat(\merges)$;
    similarly, whenever this string is only compressed by two symbols, the \maxsatacron disjunction will not be satisfied.
    As our condition assumes a compression of at least $2\satnclauses+\minclauses$ symbols, we know that we have at least $\minclauses$ strings for which a literal has an associated merge.
    We can thus write:
    \begin{subequations}
    \begin{align}
         2\minclauses \leq
        &\max_{\merges \in \mergeset^*} \sum_{\characters \in \dataset_3} |\characters| - |\bottomuptokenfull|  \\
        =
        &\max_{\merges \subset \mergeset^*} \,\,\,\,\, \sum_{\charstring{\spacesymbol\satliteral_i^1\spacesymbol\satliteral_i^2\spacesymbol} \,\in\, \dataset_3} 2\,
        \one\left\{
        \begin{array}{c}
             (\satliteral_i^1 = \satvar_j) \,\,\mathtt{and}\,\, (\mergesthreecharstrue \in \merges)  \\
               \mathtt{or} \\
             (\satliteral_i^1 = \neg\satvar_j) \,\,\mathtt{and}\,\, (\mergesthreecharsfalse \in \merges)  \\
               \mathtt{or} \\
             (\satliteral_i^2 = \satvar_{j'}) \,\,\mathtt{and}\,\, (\mergesthreecharstrueprime \in \merges)  \\
               \mathtt{or} \\
             (\satliteral_i^2 = \neg\satvar_{j'}) \,\,\mathtt{and}\,\, (\mergesthreecharsfalseprime \in \merges)
        \end{array} \right\}  \\
        =&\max_{\satvals \in \{0,1\}^{\satnvariables}} \sum_{i=1}^{\satnclauses} 2\one\{\satliteral_i^1 \lor \satliteral_i^2\} \\
        &\qquad\qquad \implies\,\, \maxsatfull
    \end{align}
    \end{subequations}
    We thus know that, if a satisfying tokenisation solution exists, then the associated \maxsatacron problem will also be satisfiable.
    This concludes the proof.
    \end{proof}

\end{document}